\documentclass[journal, 11pt, onecolumn, draftcls]{IEEEtran}
\usepackage[utf8]{inputenc}
\usepackage{cleveref}
\usepackage{amsmath}
\usepackage{enumerate}
\usepackage{amsfonts}
\usepackage{amssymb}
\usepackage{amsthm}
\usepackage{mathtools}
\usepackage{tikz}
\usepackage{pgfplots}
\usepackage{soul}
\usepackage{cite}
\usepackage{bm}
\usepackage{csquotes}
\usepackage[colorlinks=false,linktocpage=true]{hyperref}
\usetikzlibrary{shapes.geometric}
\usetikzlibrary{shapes,patterns}
\usetikzlibrary{arrows,positioning,automata,calc}
\usetikzlibrary{plotmarks}
\usepackage[caption=false]{subfig}
\usepackage{graphicx}

\newtheorem{remark}{Remark}
\newtheorem{example}{Example}

\newtheorem{theorem}{Theorem}
\newtheorem{lemma}{Lemma}

\newcommand{\F}{\mathbb{F}}

\newcommand\mydef{\coloneqq}

\DeclareMathOperator{\RS}{RS}

\newcommand{\FF}{\F}
\newcommand\bfA{\bm{A}}
\newcommand\bfB{\bm{B}}
\newcommand\bfC{\bm{C}}
\newcommand\bfD{\bm{D}}
\newcommand\bfE{\bm{E}}
\newcommand\bfM{\bm{M}}
\newcommand\bfQ{\bm{Q}}
\newcommand\bfR{\bm{R}}
\newcommand\bfS{\bm{S}}
\newcommand\bfT{\bm{T}}

\newcommand\bfX{\bm{X}}

\newcommand\bfa{\bm{a}}
\newcommand\bfb{\bm{b}}
\newcommand\bfe{\bm{e}}
\newcommand\bfu{\bm{u}}
\newcommand\bfv{\bm{v}}
\newcommand\bfw{\bm{w}}
\newcommand\bfx{\bm{x}}
\newcommand\bfy{\bm{y}}

\newcommand\bfgamma{\bm{\gamma}}
\newcommand\bfPsi{\bm{\Psi}}

\newcommand\bfzero{\bm{0}}

\newcommand\calA{\mathcal{A}}
\newcommand\calC{\mathcal{C}}

\newcommand\bcdot{\boldsymbol{\cdot}}

\tikzstyle{zcorner} = [draw, ellipse, minimum size=10pt, align=center]

\pgfplotsset{compat=1.14}
\begin{document}

\title{Private Information Retrieval Schemes with Regenerating Codes}

\author{
  \IEEEauthorblockN{
    Julien Lavauzelle\IEEEauthorrefmark{1}\IEEEauthorrefmark{2},
    Razane Tajeddine\IEEEauthorrefmark{1}\IEEEauthorrefmark{4},
    Ragnar Freij-Hollanti\IEEEauthorrefmark{4},
    Camilla Hollanti\IEEEauthorrefmark{4}}\\
  \IEEEauthorblockA{\IEEEauthorrefmark{2}
    Laboratoire LIX, École Polytechnique, Inria \& CNRS UMR 7161, University Paris-Saclay,
    Palaiseau, France\\
    Email: julien.lavauzelle@inria.fr
  }\\
  \IEEEauthorblockA{\IEEEauthorrefmark{4} Department of Mathematics and Systems Analysis, 
    Aalto University School of Science, 
    Espoo, Finland\\
    Emails: \{razane.tajeddine, ragnar.freij, camilla.hollanti\}@aalto.fi}\\
  \thanks{\IEEEauthorblockA{\IEEEauthorrefmark{1}: Both authors contributed equally to this manuscript.}}
}

\maketitle

\begin{abstract}
%\RED{Cami: Write in passive tense. Regenerating code is by definition optimal in terms of the storage-bandwidth tradeoff, so dont say "optimal reg code" but just "reg code" (I tuned the definition below accordingly). Among all reg codes, the best tradeoff is achieved when $d=n-1$, so if one says optimal reg code, it should refer to this case, which it doesn't in our case (we consider any $d$ for which PM can be constructed).}
%\BLUE{Julien: From what I read before (e.g. Rashmi et al.), regenerating codes are not necessarily optimal in terms of storage-bandwidth tradeoff. But MSR and MBR codes are, of course.}
%\GREEN{Ragnar: I would also prefer it if the term "regenerating codes" were not reserved for optimal ones. But of course, if it is standard in the literature, let's stick to that.}\RED{cami: Let's stick with the original definition of Dimakis, that is also how it's mostly used. Rashmi makes an exception.}

  A private information retrieval (PIR) scheme allows a user to retrieve a file from a database without revealing any information on the file being requested. As of now, PIR schemes have been proposed for several kinds of storage systems, including replicated and MDS-coded data. In this paper, the problem of constructing a PIR scheme on regenerating codes is considered. 
  
  A regenerating code is a storage code whose codewords are distributed among $n$ nodes, enabling efficient storage of files, as well as low-bandwidth retrieval of files and repair of nodes. In this work, a PIR scheme on regenerating codes is constructed, using the product-matrix (PM) framework of Rashmi, Shah and Kumar. Both the minimum-bandwidth (MBR) and minimum-storage (MSR) settings are considered, and the structure given by the PM framework is used in order to reduce the download communication complexity of our schemes.
\end{abstract}

\section{Introduction}

Private information retrieval (PIR) allows a user to retrieve a file from a storage system without revealing what file she is interested in. The problem of constructing PIR schemes was introduced by Chor, Goldreich, Kushilevitz and Sudan~\cite{PIR1995, chor1998private}, where data was considered to be replicated on multiple servers. In the first model, it was assumed that the data is a bitstring $\bfx \in \{ 0,1 \}^m$, and the user would like to retrieve a bit $x^f$ without revealing the index $f$ to the servers. Since its introduction, much work has been done on the replicated data model~\cite{SunJ17, SunJ18a, yekhanin2010private, beimel2001information, beimel2002breaking}. The asymptotic capacity for a PIR scheme over a storage system where the files are replicated on $n$ servers was found to be $1-1/n$ \cite{SunJ17}.

On the other hand, there is a lot of interest in using codes for storage in order to minimize storage overhead. As a consequence, many works also considered the PIR model where the data is not replicated but coded and distributed over multiple servers, see \emph{e.g.} \cite{shah2014one, chan2014private, BanawanU18, fazeli2015pir, blackburn2016pir,tajeddine2016private, FreijGHK17,kumar2017achieving, freij2018t}. The asymptotic capacity for a PIR scheme for a storage system where the files are coded on multiple servers using an $[n,k]$ MDS code was found to be $1-k/n$ \cite{BanawanU18}. The present work will focus on the case of \emph{regenerating codes} as storage codes.

Regenerating codes are a class of codes dedicated to distributed storage, achieving the optimal tradeoff between the bandwidth needed for a node repair and the amount of data each node needs to store. These codes were pioneered by Dimakis \emph{et al.}~\cite{Dimakis} who notably produced a cut-set bound on the parameters of the codes. This bound materializes two interesting optimal settings: one for which the repair communication cost in minimized, called the minimum-bandwidth regenerating (MBR) point, and one for which the nodes store the least data, called the minimum-storage regenerating (MSR) point. Rashmi \emph{et al.}~\cite{rashmi2011optimal} then proposed optimal constructions for these two specific settings, based on the so-called \emph{product-matrix} (PM) framework. Many other works followed, including \cite{SuhR11, KamathSPRLKKV13, RavivSE16} for the construction of MBR/MSR codes. Also notice that security against eavesdroppers in regenerating codes have been intensively studied, \emph{e.g.} in \cite{ShahRK11, PRK11}.

In this paper, we propose PIR schemes for the optimal PM constructions of Rashmi \emph{et al.}~\cite{rashmi2011optimal} in both MBR and MSR settings. The protocols we give use the symmetry and the redundancy inherent to the PM constructions, in order to decrease the number of symbols downloaded from the servers. As a consequence, we outperform the very recent constructions of PIR schemes over PM codes given by Dorkson and Ng in~\cite{DorksonN18, DorksonN18b}, which represent the only existing works on PIR schemes for MBR/MSR codes, to the best of our knowledge.

Concerning PM-MBR codes, we obtain a PIR rate strictly larger than $1 - \frac{k}{n}$, where $n$ is the total number of servers and $k$ is the smallest number of servers it is necessary to contact in order to retrieve a file in a regenerating code. This can be compared to the capacity of \emph{scalar} $[n, k]$ MDS-coded PIR schemes for an unbounded number of messages, which is exactly $1 - \frac{k}{n}$~\cite{chan2014private, BanawanU18}. Thus, this presents another incentive to use MBR codes for storage systems. It is important to note that, though our result might seem contradictory, PM codes are \emph{vector} codes, hence the bound in~\cite{chan2014private, BanawanU18} does not apply. In this work, the PIR rate we obtain remains below $1 - \frac{k}{n} + \frac{k(k-1)}{2nd}$, which can be considered as an upper bound on the capacity of PIR schemes based on $(n, k, d)$ MBR codes\footnote{Indeed, under the constraint $\beta = 1$, a PM-MBR code is an $[nd, B]$ linear code over $\FF_q$, where $B = kd - \frac{k(k-1)}{2}$. Moreover it is known that $1 - \frac{B}{nd}$ is an upper bound on the PIR capacity of an $[nd, B]$ linear code with such parameters, since it is the capacity of an $[nd,B]$ MDS code~\cite{BanawanU18}.}.

In the PM-MSR setting, we construct a PIR scheme similar to the scheme in the PM-MBR setting, where we consider $d=2k-2$ for simplicity. The PIR scheme achieves a PIR rate which is between $1-d/n$, the rate obtained by Dorkson and Ng~\cite{DorksonN18} which is also the PIR capacity of an $[n,d]$ MDS code, and $1-k/n$, the PIR capacity of an $[n,k]$ MDS code. %\BLUE{[explain our result here]}.

\section{Preliminaries}

\subsection{Notation and definitions}
\label{subsec:notation}

For $\bfa, \bfb \in \F_q^n$, we denote their inner product by $\langle \bfa , \bfb \rangle \mydef \sum_{i=1}^n a_i b_i\in\F_q$ and their component-wise (star) product by $\bfa \star \bfb \mydef (a_1 b_1, \dots, a_n b_n) \in \FF_q^n$. For $I \subset [1,n]$, we denote by $\bfa_{|I}$ the tuple obtained by restricting $\bfa$ to coordinates in $I$.  The \emph{Reed-Solomon code} of dimension $k$ with distinct evaluation points $\bfx = (x_1, \dots, x_n)$, where $x_i \in \F_q$, is defined by
\[
\RS_k(\bfx) \mydef \{ (f(x_1), \dots, f(x_n)), f \in \FF_q[X], \deg f \le k-1 \} \subseteq \FF_q^n\,.
\]
It is well-known that for any $1 \le k \le n$, the code $\RS_k(\bfx)$ is maximum-distance separable (MDS), and that $\RS_j(\bfx) \subseteq \RS_k(\bfx)$ for every $j \le k$. Therefore there exists a basis $\Gamma = \{ \bfgamma_1, \dots, \bfgamma_k \}$ of $\RS_k(\bfx)$, such that, for every $j \le k$ and every subset $I \subset [1,n]$ for cardinality $|I| \ge j$, the family $\Gamma^{(I,j)} \mydef \{ (\bfgamma_1)_{|I}, \dots, (\bfgamma_j)_{|I} \}$ is a basis of $\RS_j(\bfx_{|I}) \subseteq \FF_q^{|I|}$. For instance, one can take a degree-ordered monomial basis, explicitly given by $\bfgamma_j \mydef (x_1^j, \dots, x_n^j) \in \FF_q^n$.

Throughout this paper, we will refer to the asymptotic PIR capacity simply as the PIR capacity, as this is the only definition of PIR capacity we consider.

The \emph{Vandermonde matrix} with distinct basis elements $\bfx \in \F_q^n$ is the $n \times k$ matrix $\bfPsi \in \F_q^{n \times k}$ such that $\Psi_{i,j} = x_i^j$ for $1 \le i \le n$ and $1 \le j \le k$. We know that $\bfPsi$ generates the code $\RS_k(\bfx)$ by columns. More precisely, these columns form the monomial basis we mentioned earlier.

The nomenclature used in this paper is summarized in the following table.

\begin{table}[htb]
\caption{NOMENCLATURE} \label{tab:title} 
\vspace{-1em}
\begin{center}
\begin{tabular}{|c|p{10cm}|}
\hline
    $\calC$ & Regenerating code \\\hline
    $F$ & Number of files\\\hline
    $n$ & Number of servers\\\hline
    $k$ & Reconstruction parameter of the regenerating code\\\hline
    $d$ & Repair parameter of the regenerating code\\\hline
    $B$ & Number of symbols in a regenerating codeword\\\hline
    $\alpha$ & Storage capacity of a single server \\\hline
    $\beta$ & Repair-bandwidth of a single server \\\hline
    $\bfX = (\bfX^1, \dots, \bfX^F)$ & Set of files (database) \\\hline
    $\bfX^{f_0}$ & Specific file requested by the user\\\hline
    $\bfM^f$ & Redundant arrangement of file $\bfX^f$ in a matrix, as in the PM framework\\\hline
    $\bfC^f$ & Regenerating codeword associated to $\bfX^f$, as stored on the DSS\\\hline
    $\bfC^f[\bcdot, \bcdot, s]$ & $s$-th stripe of codeword $\bfC^f$ \\\hline
    $\bfC^f[i, \bcdot, s]$ & Sub-array of $\bfC^f[\bcdot, \bcdot, s]$ stored by $i$-th server\\\hline
    $\bfC^f[i, j, s]$ & $j$-th $\FF_q$-symbol of sub-array $\bfC^f[i, \bcdot, s]$ \\\hline
    $\bfQ_\ell $ & $\ell$-th query sent to servers\\\hline
    $R$ & Rate of a PIR scheme\\\hline
    $H(\cdot)$ & Entropy function\\\hline
\end{tabular}
\end{center}
\vspace{-1em}
\end{table}
% \GREEN{Ragnar: Should we include $\bfR$ for responses in the table?}

% \textcolor{blue}{[Julien: the following could be useless]}
% If $C = (C_{i,j}^f)$, for $1 \le i \le n, 1 \le j \le \alpha, 1 \le f \le m,$ is a $3$-dimensional data structure, we adopt a convenient notation for representing its substructures. For instance, matrix $(C^f_{i,j})_{\substack{1 \le i \le n\\ 1 \le j \le \alpha}}$ is denoted $C^f_{\bullet,\bullet}$, meaning that $\bullet$ represents a varying index and $f$ is fixed. Similarly, $C^\bullet_{i,j}$ is the $1$-dimensional vector $(C^f_{i,j})_{1 \le f \le m}$.

\subsection{Private information retrieval}

Consider a scheme between a user and $n$ servers storing an encoded version of $F$ files $\bfX^1, \dots, \bfX^F$. In the scheme, \emph{queries} $\bfQ[1], \dots, \bfQ[n]$ are sent to servers, which in return compute \emph{responses} $\bfR[1], \dots, \bfR[n]$ accordingly. Now, assume the user wants to retrieve a specific file $\bfX^{f_0}$, for $1 \le f_0 \le F$. We say the scheme achieves information-theoretic PIR against non-colluding servers, if the following requirements hold:
\[
\begin{array}{ll}
  \text{Privacy: } & H(f_0 \mid \bfQ[i] ) = H(f_0), \quad i=1,\dots,n.\\
  \text{Recovery: } & H(\bfX^{f_0} \mid \bfR[1], \dots, \bfR[n] ) = 0\,.
\end{array}
\]
Here, $H(\cdot)$ denoted the entropy function. Concerning the recovery constraint, it is also desirable that the user is able to reconstruct $\bfX^{f_0}$ explicitly from $\bfR[1], \dots, \bfR[n]$. Finally, we define the (download) PIR rate of a scheme by $R \mydef \frac{|\bfX^{f_0}|}{\sum_i |\bfR[i]|}$ where $|\cdot|$ represents the bitsize of a vector. The PIR capacity is the maximum achievable PIR rate.

\subsection{Regenerating codes}

Regenerating codes were introduced by Dimakis~\emph{et al.}\ in the context of distributed storage~\cite{Dimakis}. In an $(n, k, d, B, \alpha, \beta)$ regenerating code, a coded version of a file of size $B$ is stored on $n$ servers (or nodes), each storing $\alpha$ symbols. Besides, two additional constraints are required. The first is to give any external user the ability to retrieve the file by contacting any subset of $k$ servers. The second is to allow repair of any failed server by contacting any subset of $d \ge k$ servers and downloading $\beta$ symbols from each, \emph{i.e.}, $\gamma \mydef \beta d$ symbols in total. Parameters of regenerating codes are sometimes shortly denoted $(n,k,d)$, but one should take care that $d$ is \emph{not} the minimum distance of the code, and $k$ is \emph{not} the dimension of the code. 

Dimakis \emph{et al.}~\cite{Dimakis} proved that any storage (erasure) code must satisfy the so-called \emph{cut-set bound}
\begin{equation}
  \label{eq:bound-regenerating-codes}
B \le \sum_{i=0}^{k-1} \min\{\alpha, (d-i)\beta\}\,,
\end{equation}
and codes achieving this bound are called \emph{regenerating codes}. Dimakis \emph{et al.}\ also showed that equality in~\eqref{eq:bound-regenerating-codes} defines a tradeoff between parameters $\alpha$ and $\gamma = \beta d$, which cannot be minimized simultaneously. Optimal codes minimizing $\gamma = \beta d$ reach the minimum-bandwidth regeneration (MBR) point, while those minimizing $\alpha$ attain the minimum-storage regeneration (MSR) point.

\subsection{Product-Matrix constructions}

 In this work, we focus on the regenerating codes built by Rashmi~\emph{et al.}\ in~\cite{rashmi2011optimal}, through the \emph{product-matrix} (PM) framework. In their constructions, the authors set $\beta = 1$ without loss of generality, since regenerating codes with $\beta \ne 1$ can be built by \emph{striping} files in regenerating codes with $\beta = 1$. Therefore, for convenience we also consider the setting $\beta=1$ in what follows.

\subsubsection{PM codes in the MBR setting}

%\BLUE{[Julien: I think that, \emph{for the non-colluding case,} the PIR scheme we propose works over regenerating codes based on \emph{any} MDS code. However for the colluding case, RS codes seems to behave in a better way. Whence my question: which setting should be used?]} \textcolor{green!50!black}{You are right, depends surely on whether we want collusion or not.} 

At the MBR point with $\beta = 1$, we have the following constraints on the parameters:
\begin{equation*}
  \label{eq:parameters-MBR}
\alpha = d \quad\text{ and }\quad B = k(d-k) + \frac{k(k+1)}{2}\,.
\end{equation*}

The construction of Rashmi \emph{et al.}~\cite{rashmi2011optimal} can be presented as follows. Firstly, file (message) symbols are arranged in a $d \times d$ matrix
\begin{equation}
  \label{eq:matrix-MBR}
  \bfM = 
  \begin{pmatrix}
    \bfS    & \bfT \\
    \bfT^\top  & \bfzero
  \end{pmatrix}
\end{equation}
where $\bfS$ is a $k \times k$ symmetric matrix containing $\frac{k(k+1)}{2}$ distinct file symbols, and $\bfT$ is a $k \times (d-k)$ matrix containing the remaining $k(d-k)$ file symbols. Let now $\bfPsi$ be an $n \times d$ Vandermonde matrix over a large enough finite field $\F_q$. The code is defined as $\calC\mydef\bfPsi\bfM\in\FF_q^{n\times d}$. 
%\RED{Cami: what is $M^f$? Define.}
%The encoding of file $X^f$ is defined as $\bfC^f \mydef \bfPsi \bfM^f \in \FF_q^{n \times d}$.
% \RED{cami: ``for some fixed $m_i$''. (This equation defines a whole codebook when varying $m_i$, i.e., we can define many files w this. Explain more precisely. Denote a specific file (codeword) by $\mathbf{c}\in \mathbf{C}$). Take note on later instances as well!}.
The $j$-th row of a codeword in $\calC$
% \RED{$\mathbf{c}\in$} $\bfC$ \BLUE{[the codeword is denoted $\bfC$, see above]}
is stored on server $S_j$, for $j=1, \dots, n$, and contains at most $\alpha = d$ information symbols.
%Throughout the paper, the above construction is referred to as a PM-MBR code, and denoted by $\calC =\bfPsi \bfM \subseteq \FF_q^{n \times d}$.
%\RED{You just defined it as $\bfC$, why do you suddenly change to $\mathcal{R}$? Then later you still/again talk about $\bfC$.}
%\BLUE{In my mind, $\bfC$ is a codeword of the regenerating code $\calR$. It is defined by the equation $\bfC = \bfPsi \bfM$, where $\bfM$ is an arrangement of the symbols of the original file. But maybe it was not clear enough.}.
%\GREEN{Ragnar: I have no problem reading the presentation above, such that $C$ is a codeword in a code $\calR$, but I'm not sure why we use this notation. If we do, then $\calR$ should be included in the nomenclature table.}
%\BLUE{[Julien: Now, the regenerating code is $\calC$].}
Notice that $\calC$ is an $[nd, B]$ linear code over $\FF_q$. For clarity, let us now rewrite the example given by the authors in~\cite[Sec. IV.A.]{rashmi2011optimal}.
\begin{example}[Optimal PM-MBR code]
  \label{ex:MBR}
  Consider the setting $(n, k, d) = (6, 3, 4)$ over the field $\F_7$. The original file contains $B = k(d-k)+\frac{k(k+1)}{2} = 9$ symbols. Let $\bfx = (1, 2, 3, 4, 5, 6) \in \FF_7^6$. The generator (Vandermonde) matrix and the message matrix are then given as:
  \[
  \bfPsi = 
  \begin{pmatrix}
    1 & 1 & 1 & 1 \\
    1 & 2 & 4 & 1 \\
    1 & 3 & 2 & 6 \\
    1 & 4 & 2 & 1 \\
    1 & 5 & 4 & 6 \\
    1 & 6 & 1 & 6 \\
  \end{pmatrix}
  ,\;
  \bfM=\begin{pmatrix}
    m_1 & m_2 & m_3 & m_7 \\
    m_2 & m_4 & m_5 & m_8 \\
    m_3 & m_5 & m_6 & m_9 \\
    m_7 & m_8 & m_9 & 0   \\
  \end{pmatrix}.
  \]
\end{example}

\subsubsection{PM codes in the MSR setting}
\label{subsubsec:def-MSR}

In the MSR setting with $\beta = 1$, parameters $\alpha$ and $B$ are given by:
\begin{equation*}
  \label{eq:parameters-MSR}
\alpha = d-k+1 \quad\text{ and }\quad B = k(d-k+1)\,.
\end{equation*}

In~\cite{rashmi2011optimal}, the authors construct PM codes at the MSR point, for $d \geq 2k -2$. In this setting, $d \leq 2\alpha$ and $B \leq \alpha(\alpha+1)$. In this work, for simplicity, we assume %the case where
$d = 2k -2$ as it is the case for the first construction given in~\cite{rashmi2011optimal}. Thus, $d$ and $B$ can be simplified as $d = 2\alpha$ and $B = \alpha(\alpha+1)$. Note that the scheme we propose further in Section~\ref{sec:PIR-MSR} can be easily generalized to the case where $d \geq 2k -2$.

File symbols are  arranged in a $2\alpha \times \alpha$ matrix
\[
\bfM = 
\begin{pmatrix}
  \bfS_1 \\
  \bfS_2
\end{pmatrix}
\]
where each $\bfS_i$ is an $\alpha \times \alpha$ symmetric matrix containing $\frac{\alpha(\alpha+1)}{2}$ file symbols.  Let $\bfPsi$ be an $n \times 2\alpha$ Vandermonde matrix over $\FF_q$. As in the MBR setting, the $j$-th row of a codeword from the code $\calC\mydef \bfPsi \bfM$ is stored on server $S_j$, for $j=1, \dots, n$.

This construction is referred to as PM-MSR codes. Let us also rewrite the example given in~\cite[Sec. V.A.]{rashmi2011optimal}.
\begin{example}[Optimal PM-MSR code]
  \label{ex:MSR}
  Consider the setting $(n, k, d) = (6, 3, 4)$ over $\F_{13}$, which gives the file size $B = 6$. Let $\bfx = (1, 2, 3, 4, 5, 6) \in \FF_{13}^6$. Matrices $\bfPsi$ and $\bfM$ are then given by:
  \[\bfPsi = 
  \begin{pmatrix}
    1 & 1 &  1 &  1 \\
    1 & 2 &  4 &  8 \\
    1 & 3 &  9 &  1 \\
    1 & 4 &  3 & 12 \\
    1 & 5 & 12 &  8 \\
    1 & 6 & 10 &  8 \\
  \end{pmatrix}
  ,\; \bfM=
  \begin{pmatrix}
    m_1 & m_2  \\
    m_2 & m_3  \\
    m_4 & m_5  \\
    m_5 & m_6  \\
  \end{pmatrix}.
  \]
\end{example}

\section{A PIR scheme in the MBR setting}

In this section, we consider a PM-MBR code $\calC$
%\RED{cami: Ok I guess you define M for fixed numbers, not variables, and then C is your codeword, and R is your code.
%\BLUE{[That's it.]}
%This is pretty non-standard. Why not just define C as the code w variables $m_i$, and no need to define a code R.} \BLUE{Well, actually the problem is that I needed to make a difference between the whole regenerating code (of length $n \times d$) and the \enquote{column-code} which is basically a (subcode of a) RS code of length $n$. But again, maybe the notation was awkward.}
%\GREEN{Ragnar: Couldn't we then just use $C^S$ for the column code (now denoted $\calR^S$)? Or do I misunderstand something?}
over $\FF_q$, with parameters $(n, k, d)$. Recall that $\calC$ is also a linear code over $\FF_q$ of length $nd$ and dimension $B = k(d-k) + \frac{k(k+1)}{2}$.

\subsection{System setup}

We consider a database $\bfX$ composed of $F$ files $\bfX^1, \dots, \bfX^F$, such that each $\bfX^f$ consists of $B = k(d-k) + \frac{k(k+1)}{2}$ information symbols. For every $1 \le f \le F$, the symbols of file $\bfX^f$ are subdivided into $S \ge 1$ \emph{stripes} (or subdivisions) and organized in a $3$-dimensional array $\bfM^f$ (that we abusively name a \emph{matrix}), such that 
\[
\bfM^f = \left( M^f[i,j,s],
  {
    \scriptsize
    \begin{array}{l}
      1 \le i \le d \\
      1 \le j \le d\\
      1 \le s \le S
    \end{array}
  }
\right) \in \mathbb{F}_q^{d \times d \times S},
\]
where for every $i, j, s, f$, we have $M^f[i,j,s] \in \F_q$. Following the PM framework, every stripe $\bfM^f[\bcdot, \bcdot, s]$ must the form given in~\eqref{eq:matrix-MBR}. Also notice that, by construction of the regenerating code $\calC$, for all $i, j, s, f$, we have:
\[
 M^f[i, j, s] = M^f[j, i ,s]\,,
\]
and 
\[
M^f[i, j, s] = 0 \;\; \text{ if } i \ge k+1 \text{ and } j \ge k+1\,.
\]
We also use the notation $\bfM \mydef (\bfM^1, \dots, \bfM^F)$.

For every $j, s, f$, the column $\bfM^f[\bcdot, j, s] \in \FF_q^d$ is encoded using a Reed-Solomon code $\RS_d(\bfx)$, resulting in a codeword
\[
 \bfC^f[\bcdot, j, s] = \sum_{r=1}^d M^f[r, j, s] \bfgamma_r,
\]
where we recall that $\Gamma = \{ \bfgamma_1, \dots, \bfgamma_d \}$ denotes a suitable basis for sequences of Reed-Solomon codes (see Section~\ref{subsec:notation}). Due to the form of message matrices $\bfM^f$, one can also remark that $\bfC^f[\bcdot, j, s] \in \RS_k(\bfx)$ if $j \ge k+1$.

\begin{figure}
  \centering
  \includegraphics{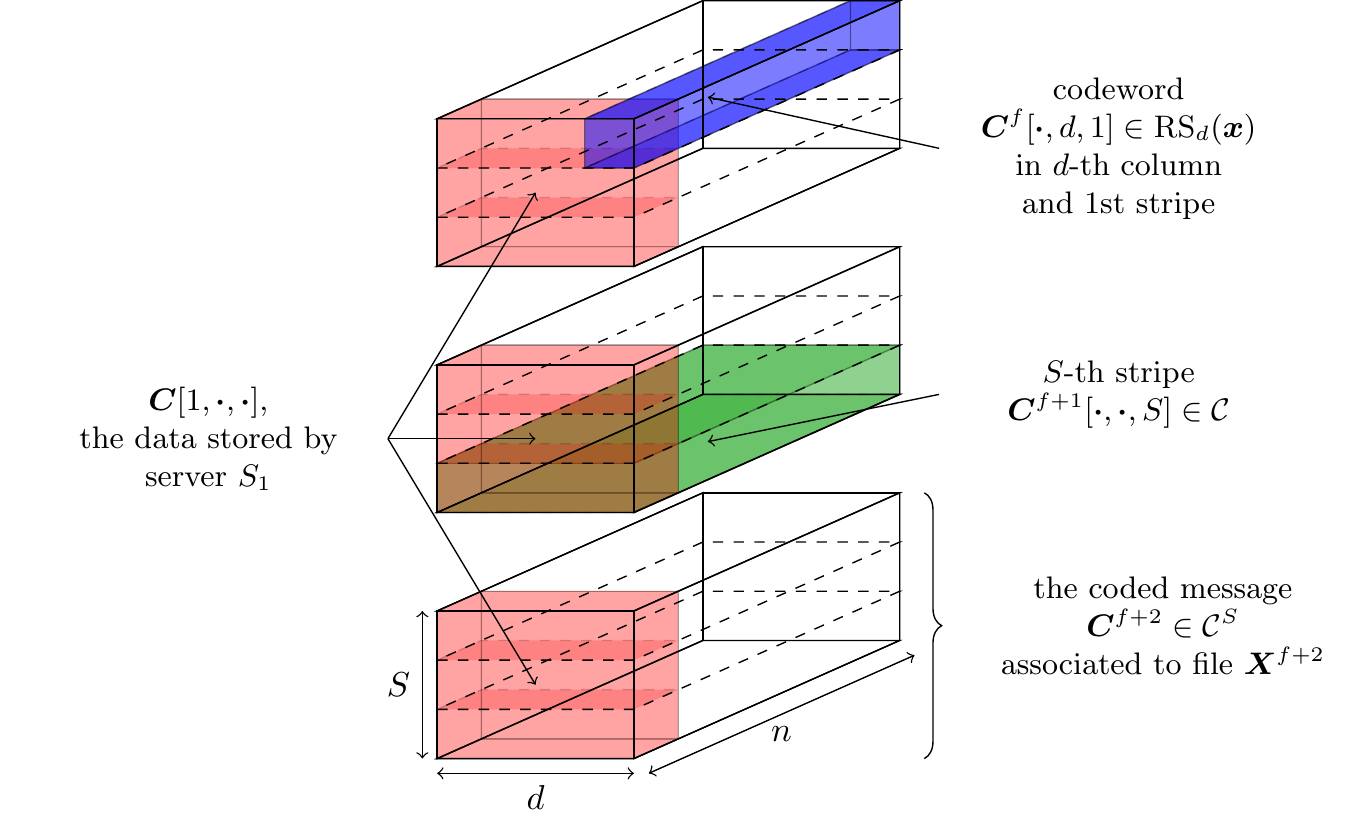}
  \caption{An illustration of the arrangement of files, stripes and codewords in the storage system. A system of $n$ servers stores encoded files represented by $S \times d \times n$ cuboids (in the figure, only three of them are represented). Foreground (red) blocks represent data stored by the first server. The horizontal block (in green) in the middle cuboid represents a stripe, which lies in the regenerating code $\calC$. Top right block (in blue) is a column of a stripe, and typically lies in an MDS code.}
  \label{fig:illustration-storage}
\end{figure}

\subsection{Intuition}

The idea behind the constructed PIR scheme is to use the symmetric property of matrices $\bfM^f$ as a way to reuse information in order to decrease the download complexity of the scheme. We note that the servers are assumed not to collude. In this scheme, each file is divided into $S = n-k$ stripes. The user generates a set of $k$ queries to the servers, similarly to the scheme in \cite{tajeddine2016private}. A query is defined as an $n \times S\times F$ vector that is sent by the user to retrieve information. Randomness is embedded in the queries as a way to hide the requested file's identity, in a similar manner to one-time padding. Naturally, if privacy were not a concern, a query to retrieve file $\bfX^{f_0}$ would be the vector of size $n \times S\times F$ with zeroes everywhere, except in positions $f = f_0$ corresponding to the requested file. 

%The queries are then sent to the servers such that, in every query, exactly one piece from each stripe is requested, and every stripe is requested from a server exactly one time in all queries \BLUE{[Julien: sorry I did not get the last sentence]}\GREEN{[Razan: I am not sure how to explain it better, basically saying that we want the queries to give us $k$ pieces about each stripe and we do not want redundant information. If anyone can help writing this sentence better :).]}. 
The queries are then sent to the servers which project queries on their stored data the following manner. 
%After receiving the queries, the servers project all $k$ queries on their stored data in the last $d-k$ columns, since each of these columns store $kS$ file stripes encoded using an $[n,k]$ maximum distance separable (MDS) code.
For the last $d-k$ columns, since each of these columns stores file stripes encoded using an $[n, k]$ MDS code, servers are asked to project \emph{all} the queries on the data they hold, similarly to~\cite{tajeddine2016private}.
For each of the other columns, stripes contain information already retrieved from the previously used columns, due to the nature of the product-matrix construction. Thus, from server $S_d$ down to server $S_1$, servers are asked to project  on their stored data a \emph{decreasing} subset of the initial set of queries. This still enables the user to reconstruct the requested file, due to the fact that she had peeled off some randomness and information symbols from previous columns. Moreover, it allows her to run a more efficient PIR scheme on an $[n, k']$ MDS code with where $k'<k$. More details are given in the upcoming sections.

\subsection{The PIR scheme}
\label{subsec:PIR-MBR}

In this section, we describe the PIR scheme explicitly. Let us assume that the user wants to retrieve a file $\bfX^{f_0}$, for some $1 \le f_0 \le F$. We fix the number of stripes to $S = n-k$, and we consider a $k$-tuple of queries $\bfQ = (\bfQ_1, \dots, \bfQ_k)$, such that for $1 \le \ell \le k$, query $\bfQ_\ell$ has the following form:
\[
\bfQ_\ell =  \left( Q^f_\ell[i,s],
  {
    \scriptsize
    \begin{array}{l}
      1 \le i \le n \\
      1 \le s \le S\\
      1 \le f \le F
    \end{array}
  }
\right).
\]
Notice that, since the same set of queries is meant to be used for every column, query $\bfQ_\ell$ does \emph{not} depend on a column index $j \in [1,d]$. This property is fundamental for the privacy of the scheme.

The sub-query $\bfQ^f_\ell[i,\bcdot]$ is then sent to server $S_i$, for each $1 \le i \le n$. The response $R_\ell[i, j] \in \FF_q$ of server $S_i$ with respect to the pair $(\ell,j)$, is then defined as:
\[ 
 R_\ell[i, j] \mydef \langle \bfQ_\ell[i, \bcdot], \bfC[i, j, \bcdot] \rangle = \sum_{s,f} Q^f_\ell[i, s] C^f[i, j, s]\,.
 \]
We also denote by $\bfR_\ell[\bcdot, j] \mydef (R_\ell[1,j], \dots, R_\ell[n,j]) \in \FF_q^n$.

  {\bf Generation of $\bfQ$.} The random tuple of queries $\bfQ$ is defined as the sum of two components.
\begin{enumerate}
  \item A random part $\bfD$, defined as follows. For every $\ell, s, f$, a symbol $\lambda_{\ell, s, f} \in \FF_q$ is picked uniformly at random and independently of others. Then, for every $1 \le i \le n$, we define:
    \[
    D^f_\ell[i, s] = \lambda_{\ell, s, f}.
    \]
    In other words, $\bfD^f_\ell[\bcdot, s] \in \FF_q^n$ is a word picked uniformly at random in the repetition code of length~$n$.
  \item A deterministic part $\bfE^{(f_0)}$, also called the retrieval pattern. This pattern is defined by:
    \[
    E_\ell^{(f_0),f}[i, s] = 
    \left\{
      \begin{array}{ll}
        1 &\text{if } f=f_0 \text{ and } n-i = \ell + s - 2\quad (\text{mod } S),\\
        0 &\text{otherwise.}
      \end{array}
    \right.
    \]
\end{enumerate}
Finally, the tuple of queries $\bfQ$ is defined by $\bfQ \mydef \bfD + \bfE^{(f_0)}$. Notice here that each query is sent to the servers by hiding the deterministic part with a random vector. Therefore, the privacy of the scheme still holds.%, \emph{i.e.}, $$I(f_0|\bfQ^f_\ell[i,j,\bcdot) =0.$$ 
\medskip

  {\bf Server responses to queries.} We now assume that $\bfQ[i, \bcdot]$ is sent to server $S_i$, for every $1 \le i \le n$. In the proposed protocol, the set of responses required by the user depend on the index $j \in [1,d]$ of the column, as described below:
  \begin{itemize}
    \item For columns $k+1 \le j \le d$,  every server $S_i$, $1 \le i \le n$, must send back to the user the responses $R_\ell[i,j]$, where $1 \le \ell \le k$ .
    \item For columns $1 \le j \le k$, only servers $S_i$ such that $k-j+1 \le i \le n$ are required to respond to the user. Those servers $S_i$ must compute and send the subset of responses $R_\ell[i,j]$, such that $1 \le \ell \le j$. 
  \end{itemize}
We here emphasize that, for these first columns $1 \le j \le k$, the \emph{subset} of servers $S_i$, $i \in [k-j+1, n]$, send the \emph{subset} of responses $R_\ell[i,j]$, $\ell \in [1, j]$ to the user. This is a key point in order to achieve a good PIR rate --- see Example~\ref{ex:PIR-MBR} for an illustration.
\medskip

  {\bf Reconstruction of $\bfX^{f_0}$.} The recovery is run columnwise, from column $d$ down to column $1$. For each step $j$, $1 \le j \le d$, the goal is to retrieve $\bfM^{f_0}[\bcdot, j, \bcdot]$ along with some random vectors.
  \begin{itemize}
    \item \emph{For $k+1 \le j \le d$.} A precise description of the recovery algorithm is given in the proof of Lemma~\ref{lem:correctness-last-columns}. In short, it consists of running, \emph{independently on each column $\bfC[\bcdot, j, \bcdot]$}, the reconstruction of the PIR scheme over an MDS code described in \cite{tajeddine2016private}. Indeed, each $\bfC[\bcdot, j, \bcdot]$ can be viewed as a smaller database encoded and stored in an $[n, k]$ MDS storage system. This procedure allows the user to recover $\bfM^{f_0}[\bcdot, j, \bcdot]$, but one should notice that she can also collect random vectors $\sum_{s,f} M^{f}[r,j,s] \bfD^f_\ell[\bcdot, s] \in \FF_q^n$, for all $1 \le r, \ell \le k$.
    \item \emph{For $1 \le j \le k$.} At step $j$, we can assume that for every $j' \ge j+1$, the user has already collected 
\begin{itemize}
  \item $\bfM^{f_0}[\bcdot, j', \bcdot]$ and
  \item the random vectors $\sum_{s,f} M^f[r,j',s] \bfD^f_\ell[\bcdot, s] \in \FF_q^{n-k+\min\{k,j'\}}$ for every $1 \le r,\ell \le \min\{k, j'\}$. 
\end{itemize}
Recall that $M^f[r,j',s] = M^f[j',r,s]$ and that every $\bfD^f_\ell[\bcdot, s]$ lies in a repetition code. As a consequence, the user knows $\sum_{s,f} M^f[r,j,s] \bfD^f_\ell[\bcdot, s] \in \FF_q^{n-k+\min\{k,j'\}}$ for every $j+1 \le r \le d$ and every $1 \le \ell \le j$. The retrieval process described in the proof of Lemma~\ref{lem:correctness-first-columns} then ensures that the user can retrieve $\bfM^{f_0}[\bcdot,j,\bcdot]$ and the random vectors $\sum_{s,f} M^f[r,j,s] \bfD^f_\ell[\bcdot, s] \in \FF_q^{n-k+j}$ for  every $1 \le r,\ell \le j$. %Informally, the idea the following. The randomness the user has collected from previous rounds/steps allows her to project the servers' responses to a smaller space, by peeling off the known information, which allows her to run a more efficient PIR scheme. \BLUE{[find a better sentence to explain this in a few words... it will also be useful in the abstract/introduction]}
  \end{itemize}

We start by giving a simple example before diving into technical proofs. 

\begin{example}
  \label{ex:PIR-MBR}
  We use the $(6, 3, 4)$ PM-MBR regenerating code described in Example~\ref{ex:MBR}. For this purpose, the files are divided into $S = n-k = 3$ stripes, and the user sends $k=3$ query vectors:
  \[
  \centering
  \begin{array}{|c|c|c|c|}
    \hline
                       & \text{Query } 1   & \text{Query } 2   & \text{Query } 3    \\\hline
    \text{Server } S_1 & \bfu              & \bfv              & \bfw               \\\hline
    \text{Server } S_2 & \bfu              & \bfv              & \bfw               \\\hline
    \text{Server } S_3 & \bfu              & \bfv              & \bfw               \\\hline
    \text{Server } S_4 & \bfu + \bfe_{f_0,1} & \bfv + \bfe_{f_0,2} & \bfw + \bfe_{f_0,3} \\\hline
    \text{Server } S_5 & \bfu + \bfe_{f_0,2} & \bfv + \bfe_{f_0,3} & \bfw + \bfe_{f_0,1} \\\hline
    \text{Server } S_6 & \bfu + \bfe_{f_0,3} & \bfv + \bfe_{f_0,1} & \bfw + \bfe_{f_0,2} \\\hline
  \end{array}
  \]  
  where $\bfe_{f_0,s_0} \in \FF_q^{F \times S}$ is the deterministic vector with all zeros, but one $1$ in position $(f_0,s_0)$, which corresponds to stripe $s_0$ of what is stored from file $\bfX^{f_0}$. Vectors $\bfu, \bfv, \bfw \in \FF_q^{F \times S}$ are uniformly random vectors.
  
  The servers project the data stored in columns $3$ and $4$ on all the queries. Server $S_1$ does not respond to any other queries. Servers $S_2,\dots, S_6$ project only the first $2$ queries on the data stored in their second column. Server $S_2$ does not respond to any other queries. Servers $S_3,\dots, S_6$ project only the first query on the data stored in column $1$. Then the servers send this information back to the user.
%  \GREEN{Ragnar: I think a figure here would help readability and intuition a whole lot. But that can wait until the revision round.} \BLUE{Razan: I will build a figure for this, but I think it would take time, so yes, it would be better to wait on this till the revision.}
  
  $\bullet$ Decodability: In this example $d-k=1$. For the last row, the user receives the responses from all three queries from all six servers. The storage code for the last row is a $[6,3]$ MDS code. If we look at the responses to the first query from the last column, it will be:  
  \[
  \centering
  \small
  \begin{array}{|c|c|}
    \hline
        & \text{Response } 1    \\\hline
    S_1 & \sum_{f=1}^F\sum_{s=1}^3u_{f,s}(M^f[1,4,s]+M^f[2,4,s]+M^f[3,4,s])              \\\hline
    S_2 & \sum_{f=1}^F\sum_{s=1}^3u_{f,s}(M^f[1,4,s]+2M^f[2,4,s]+4M^f[2,4,s])            \\\hline
    S_3 & \sum_{f=1}^F\sum_{s=1}^3u_{f,s}(M^f[1,4,s]+3M^f[2,4,s]+2M^f[3,4,s])            \\\hline
    S_4 & \sum_{f=1}^F\sum_{s=1}^3u_{f,s}(M^f[1,4,s]+4M^f[2,4,s]+2M^f[3,4,s]) + \textcolor{blue}{M^1[1,4,1]+4M^1[2,4,1]+2M^1[3,4,1]} \\\hline
    S_5 & \sum_{f=1}^F\sum_{s=1}^3u_{f,s}(M^f[1,4,s]+5M^f[2,4,s]+4M^f[3,4,s]) + \textcolor{blue}{M^1[1,4,2]+5M^1[2,4,2]+4M^1[3,4,2]} \\\hline
    S_6 & \sum_{f=1}^F\sum_{s=1}^3u_{f,s}(M^f[1,4,s]+6M^f[2,4,s]+M^f[3,4,s]) + \textcolor{blue}{M^1[1,4,3]+6M^1[2,4,3]+1M^1[3,4,3]} \\\hline
  \end{array}
  \]

  From the above table, we can see that the user can recover the three random symbols $$\sum_{f=1}^F\sum_{s=1}^3u_{f,s}M^f[1,4,s],$$
$$\sum_{f=1}^F\sum_{s=1}^3u_{f,s}M^f[2,4,s]$$ and $$\sum_{f=1}^F\sum_{s=1}^3u_{f,s}M^f[3,4,s],$$ along with the three required symbols $$M^1[1,4,1], M^1[2,4,2], M^1[3,4,3].$$
Following the same reasoning, from the second and third queries the user can retrieve the random symbols $$\sum_{f=1}^F\sum_{s=1}^3v_{f,s}M^f[1,4,s],  \sum_{f=1}^F\sum_{s=1}^3w_{f,s}M^f[1,4,s],$$
$$\sum_{f=1}^F\sum_{s=1}^3v_{f,s}M^f[2,4,s], \sum_{f=1}^F\sum_{s=1}^3w_{f,s}M^f[2,4,s]$$ and $$\sum_{f=1}^F\sum_{s=1}^3v_{f,s}M^f[3,4,s],\sum_{f=1}^F\sum_{s=1}^3w_{f,s}M^f[3,4,s],$$ along with the required symbols, $$M^1[1,4,2], M^1[2,4,3], M^1[3,4,1], M^1[1,4,3], M^1[2,4,1], M^1[3,4,2].$$ Notice that the PIR scheme run over the fourth column achieves a PIR rate of $3/6$.

For the third column, the storage code is a $[6,4]$ MDS code. Recall that $\bfM^f[3, 4, \bcdot] = \bfM^f[4, 3, \bcdot]$ for every $f$, and the user has already collected information in the responses from column $4$. As a consequence, the user knows the vector $\bfM^1[4,3,\bcdot]$ as well as the random symbols 
$$\sum_{f=1}^F\sum_{s=1}^3u_{f,s}M^f[4,3,s],$$
$$\sum_{f=1}^F\sum_{s=1}^3v_{f,s}M^f[4,3,s]$$
and
$$\sum_{f=1}^F\sum_{s=1}^3w_{f,s}M^f[4,3,s].$$
Therefore, the responses from the third column allow the user to decode the symbols, just like the responses from the last column. The user, thus,  recovers $\bfM^1[1,3,\bcdot]$, $\bfM^1[2,3,\bcdot]$, and $\bfM^1[3,3,\bcdot]$ with a rate $3/6$.

For the second column, the storage code is also a $[6,4]$ MDS code, but the user can use the information she collected from columns $3$ and $4$. More precisely, the user already knows vectors $\bfM^1[2,3,\bcdot]$, $M^1[2,4,\bcdot]$, and random symbols 
$$\sum_{f=1}^F\sum_{s=1}^3u_{f,s}M^f[2,3,s], \sum_{f=1}^F\sum_{s=1}^3u_{f,s}M^f[2,4,s],$$
$$\sum_{f=1}^F\sum_{s=1}^3v_{f,s}M^f[2,3,s], \sum_{f=1}^F\sum_{s=1}^3u_{f,s}M^f[2,4,s]$$
and
$$\sum_{f=1}^F\sum_{s=1}^3w_{f,s}M^f[2,3,s], \sum_{f=1}^F\sum_{s=1}^3w_{f,s}M^f[2,4,s].$$
Thus, the user does not need the response from server $S_1$ in order to decode the symbols. It means that the code can be assumed to be reduced to a $[5, 2]$ MDS code. The user can then decode the parts $\bfM^1[1, 2, \bcdot]$, and $\bfM^1[2,2,\bcdot]$ from servers $S_2, \dots, S_6$ and from the first $2$ queries, with rate $6/10=3/5$.

Following the same reasoning for the first column, the user needs only the responses of servers $S_3,\dots, S_6$ to the first query only. The storage code can be seen as a $[4,1]$ MDS code on those servers, after introducing the already known information. This allows the user to decode the last part of the file, $\bfM^1[1,1, \bcdot]$, with rate $3/4$.

Finally, the PIR rate of the scheme in this example is $R_{\rm MBR} = \frac{3+6+9+9}{4+10+18+18} = \frac{27}{50} = 0.54$. We see this rate is larger than $1 - \frac{k}{n} = 1 - \frac{3}{6} = \frac{1}{2} = 0.5$ which is the capacity of scalar MDS-coded PIR schemes, but %(hopefully) \RED{cami: What do you mean by hopefully?! Write more professionally.} 
less than $1 - \frac{B}{nd} = 1 - \frac{9}{6 \times 4} = \frac{5}{8} = 0.625$, which is an upper bound on the capacity of $[nd, B]$-coded PIR schemes.
  %[\textcolor{red}{Not sure how to write this.}] The user receives $\mathbf{u}\cdot Y_1$, $\mathbf{u}$
  
  $\bullet$ Privacy: Privacy follows from the fact that for any fixed desired file, every server gets a uniform random vector as a query. 
  
\end{example}

\subsection{Analysis}

We next prove the correctness of the PIR scheme proposed in previous section.

\begin{lemma}
  \label{lem:correctness-last-columns}
  Let $k+1 \le j \le d$. Then, conditioned on $(\bfR_1[\bcdot, j], \dots, \bfR_k[\bcdot, j])$, the following is determined:
  \begin{itemize}
  \item the piece $\bfM^{f_0}[\bcdot, j, \bcdot]$ of the desired file;
  \item the random vectors $\sum_{s,f} M^{f}[r,j,s] \bfD^f_\ell[\bcdot, s] \in \FF_q^n$ for every $1 \le r, \ell \le k$.
  \end{itemize}
\end{lemma}
\begin{proof}
  Let us fix $1 \le \ell \le k$. After receiving responses from servers, the user is able to build the response vector
  \[
  \bfR_\ell[\bcdot, j] \mydef (R_\ell[1,j], \dots, R_\ell[n,j]) \in \FF_q^n\,.
  \]
  Notice that we have 
  \[
    \bfR_\ell[\bcdot, j] = \sum_{s, f}  \bfD^f_\ell[\bcdot, s] \star \bfC^f[\bcdot, j, s] + \sum_s \bfE_\ell^{(f_0),f_0}[\bcdot, s] \star \bfC^{f_0}[\bcdot, j, s]\,.
  \]
  We can now define 
  \[
  \bfB_\ell[\bcdot, j] \mydef\sum_s \bfE_\ell^{(f_0),f_0}[\bcdot, s] \star \bfC^{f_0}[\bcdot, j, s] \in \FF_q^n\,,
  \]
  and we see that
  \begin{equation}
    \label{eq:correctness-last-columns}
    B_\ell[i, j] = 
    \left\{
      \begin{array}{ll}
        C^{f_0}[i,j,s'] & \text{if } i \ge k+1, \\
        0 & \text{otherwise,}
      \end{array}
    \right.
  \end{equation}
  where $s' \in [1, k]$ satisfies $n - i = (\ell + s' - 2 \!\mod n-k)$. In particular $\bfB_\ell[\bcdot, j]$ is supported on $[k+1, n]$, and therefore has weight at most $n-k$.
  
  Now, denote by 
  \[
  \bfA_\ell[\bcdot, j] \mydef \sum_{s, f} \bfD^f_\ell[\bcdot, s] \star \bfC^{f}[\bcdot, j, s] \in \FF_q^n\,.
  \]
  Since every $\bfD^f_\ell[\bcdot, s]$ belongs to the repetition code and $\bfC^f[\bcdot, j, s] \in \RS_k(\bfx)$, it holds that $\bfA_\ell[\bcdot, j]  \in \RS_k(\bfx)$. We have also seen that $R_\ell[i,j] = A_\ell[i,j]$ for $1 \le i \le k$, thus the user knows the $k$  first symbols of $\bfA_\ell[\bcdot, j]$. Since $[1,k]$ is an information set for $\RS_k(\bfx)$, she can recover $\bfA_\ell[\bcdot, j]$ entirely. The recovery of $\bfB_\ell[\bcdot, j]$ follows easily.
  
  Let us now recall that $\bfC^{f}[\bcdot, j, s] \in \RS_k(\bfx)$ can be written as $\sum_{r=1}^k M^{f}[r, j, s] \bfgamma_r$. Moreover, $\bfD^f_\ell[\bcdot, s]$ lies in a repetition code, hence $D^f_\ell[i, s] = \lambda_{\ell, s, f}$ for some $\lambda_{\ell, s, f} \in \FF_q$. Therefore, expressing
  \[
  \bfA_\ell[\bcdot, j] = \sum_{r=1}^d \left( \sum_{s, f} \lambda_{\ell, s, f} M^{f}[r, j, s] \right) \bfgamma_r
  \]
  in the basis $\{ \bfgamma_1, \dots, \bfgamma_d \} \subset \FF_q^n$ of nested Reed-Solomon codes $\RS_d(\bfx) \supseteq \RS_k(\bfx)$ allows us to retrieve every scalar $\sum_{s, f} \lambda_{\ell, s, f} M^{f}[r, j, s]$, or equivalently, every $\sum_{s,f} M^{f}[r,j,s] \bfD^f_\ell[\bcdot, s] \in \FF_q^n$.
  
  Finally, Equation~\eqref{eq:correctness-last-columns} shows that for every $1 \le s \le n-k$, the knowledge of $\bfB_1[\bcdot, j], \dots, \bfB_k[\bcdot, j]$ allows the user to retrieve a subset of $k$ distinct symbols of $\bfC^{f_0}[\bcdot, j, s]$, which is equivalent to retrieving $\bfM^{f_0}[\bcdot, j, s]$. Thus, she can finally obtain $\bfM^{f_0}[\bcdot, j, \bcdot]$.
\end{proof}

\begin{lemma}
  \label{lem:correctness-first-columns}
  Let $1 \le j \le k$. For every $1 \le \ell \le j$, for convenience we denote by $$\bfR_\ell[\bcdot, j] \mydef (R_\ell[k-j+1, j], \dots, R_\ell[n,j]) \in \FF_q^{n-k+j}.$$
  Then, conditioned on $(\bfR_1[\bcdot, j], \dots, \bfR_j[\bcdot, j])$ and on
  \begin{equation}
    \label{eq:random-vectors-first-columns}
  \sum_{s,f} M^f[r,j,s] \bfD^f_\ell[\bcdot,s], \quad \text{for all }\, j+1 \le r \le d, \quad 1 \le \ell \le j,
  \end{equation}
  % \begin{itemize}
  % \item $(\bfR_1[\bcdot, j], \dots, \bfR_j[\bcdot, j])$ and
  % \item $\sum_{s,f} M^f[r,j,s] \bfD^f_\ell[\bcdot,s]$ for every $j+1 \le r \le d$ and every $1 \le \ell \le j$,
  % \end{itemize}
  the following are determined:
  \begin{itemize}
  \item the piece $\bfM^{f_0}[\bcdot, j, \bcdot]$ of the desired file;
  \item random vectors $\sum_{s,f} M^f[r,j,s] \bfD^f_\ell[\bcdot, s] \in \FF_q^{n-k+j}$ for all $1 \le r, \ell \le j$.
  \end{itemize}
\end{lemma}

\begin{proof}
  Let us fix $1 \le \ell \le j$. In contrast with Lemma~\ref{lem:correctness-last-columns}, we will deal with vectors of shorter length $n-k+j$. In particular, we denote $\bfx' = (x_{k-j+1}, \dots, x_n)$. Similarly, the user is able to build the response vector $\bfR_\ell[\bcdot, j]$ of length $n-k+j$ given by
  \[
    \bfR_\ell[\bcdot, j] \mydef (R_\ell[k-j+1,j], \dots, R_\ell[n,j]) = \bfA_\ell[\bcdot, j] + \bfB_\ell[\bcdot, j]\,,
  \]
  where $\bfA_\ell[\bcdot, j]$ and $\bfB_\ell[\bcdot, j]$ are defined as in Lemma~\ref{lem:correctness-last-columns}. One can rewrite $\bfA_\ell[\bcdot, j] \in \FF_q^{n-k+j}$ as follows:
  \[
  \begin{aligned}
    \bfA_\ell[\bcdot, j]
    &= \sum_{s,f} \bfD^f_\ell[\bcdot, s] \star \bfC^f[\bcdot, j, s]\\
    &= \sum_{s,f}  \bfD^f_\ell[\bcdot, s] \star \left(\sum_{r=1}^d M^f[r, j, s] \bfgamma_r \right)\\
    &= \sum_{r=1}^j \sum_{s,f} M^f[r, j, s] \; \bfD^f_\ell[\bcdot, s] \star \bfgamma_r + \sum_{r=j+1}^d \sum_{s,f} M^f[r, j, s] \; \bfD^f_\ell[\bcdot, s] \star \bfgamma_r.
  \end{aligned}
  \]
  Therefore, using vectors in \eqref{eq:random-vectors-first-columns}
  %the second bullet in the statement of the lemma\GREEN{Ragnar: refer back with label-ref to make clear what we are conditioning on. Also, I would prefer more precise mathematical terminology, "conditioned on" rather than "with the knowledge of"}, 
  the user can build
  \[
  \bfA'_\ell[\bcdot, j] \mydef \sum_{r=j+1}^d \left( \sum_{s,f}  M^f[r, j, s] \; \bfD^f_\ell[\bcdot, s] \right) \star \bfgamma_r.
  \]
  Hence, she is able to construct
  \[
    \bfR''_\ell[\bcdot, j] \mydef \bfR_\ell[\bcdot, j] - \bfA'_\ell[\bcdot, j] = (\bfA_\ell[\bcdot, j] - \bfA'_\ell[\bcdot, j]) + \bfB_\ell[\bcdot, j]\,.
  \]
  As the basis  $\{\bfgamma_1, \dots, \bfgamma_d\}$ is ordered by degree, we see that $\bfA''_\ell[\bcdot, j] \mydef \bfA_\ell[\bcdot, j] - \bfA'_\ell[\bcdot, j]$ lies in $\RS_j(\bfx')$. Indeed, each $\{\bfgamma_1, \dots \bfgamma_j\}$ must also be a basis of smaller RS codes. Also remark that once again, the vector $\bfB_\ell[\bcdot, j] \in \FF_q^{n-k+j}$ is supported by $[k+1, n]$. Since $[k-j+1, k]$ is an information set for $\RS_j(\bfx')$, the user can thus recover $\bfA''_\ell[\bcdot, j]$ and $\bfB_\ell[\bcdot, j]$ from $\bfR''_\ell[\bcdot, j]$.
  
  Similarly to Lemma~\ref{lem:correctness-last-columns}, one can easily see that $\bfM^{f_0}[\bcdot, j, \bcdot]$ can be obtained from $\bfB_1[\bcdot, j], \dots, \bfB_j[\bcdot, j]$.

  Finally, $\bfA'_\ell[\bcdot, j]$ and $\bfA''_\ell[\bcdot, j]$ allow to reconstruct $\bfA_\ell[\bcdot, j]$. Similarly to the proof of Lemma~\ref{lem:correctness-last-columns}, the basis $\{\bfgamma_1, \dots, \bfgamma_j \}$ of $\RS_j(\bfx')$ leads to the recovery of random elements $\sum_{s,f} M^f[r,j,s] \lambda_{\ell, s, f} \in \FF_q$ for every $1 \le r, \ell \le j$.
\end{proof}

\begin{theorem}
  \label{thm:PIR-MBR}
  The scheme proposed in Section~\ref{subsec:PIR-MBR} is secure against non-colluding servers. Its PIR rate is:
  \[
  R_{\rm MBR} = \frac{3(n-k)(2d-k+1)}{6dn-3nk+3n-k^2+1}\,.
  \]
\end{theorem}
\begin{proof}
  Lemma~\ref{lem:correctness-last-columns} and Lemma~\ref{lem:correctness-first-columns} ensure that the user retrieves the correct file $\bfX^{f_0}$ as long as the servers $S_1, \dots, S_n$ follow the protocol described in Section~\ref{subsec:PIR-MBR}. Since the servers are assumed not to collude, the only way a server $S_i$ can learn information about the identity $f_0$ of the required file, is from its own query matrix $\bfQ[i,\cdot]$. Since the matrix $\bfQ[i,\cdot]$ is chosen such that it is statistically independent of $f_0$, the scheme is private. More precisely, since $\bfQ[i, \bcdot] = \bfD[i, \bcdot] + \bfE^{(f_0)}[i, \bcdot] \sim \bfD[i, \bcdot]$, we have 
% \[
% H(f_0 \mid \bfD[i, \bcdot] + \bfE^{(f_0)}[i, \bcdot]) = H(f_0 \mid \bfD[i, \bcdot]) = H(f_0)\,,
% \]
% \GREEN{Ragnar: The first equality is poorly motivated (unless we include the green text I suggested before the equation). Do we want to write
\[
H(f_0 \mid \bfQ[i, \bcdot]) = H(f_0 \mid \bfD[i, \bcdot]) = H(f_0),
\]
where $H(\cdot)$ denotes the entropy function. 

  Let us now compute the PIR rate. The file $\bfX^{f_0}$ consists of
  \[
  (n-k)B = (n-k)(k(d-k) + k(k+1)/2)
  \]
  symbols over $\FF_q$.  During step $j$, for $k+1 \le j \le d$, the user downloads $k$ responses from each server $S_1, \dots, S_n$. Hence she gets a total of $nk(d-k)$ symbols for all these steps. For columns $1 \le j \le k$, the user downloads $j$ responses from servers $S_{k-j+1}, \dots, S_n$, leading to a total of $\sum_{j=1}^k j(n-k+j)$ symbols for those steps. Therefore, we get the following PIR rate:
% \begin{eqnarray}
% R_{MBR} &=& \frac{(n-k)\left((d-k)k + k(k+1)/2\right)}{(d-k)nk + \sum_{i=0}^{k}(n-i)(k-i)}\nonumber \\
% &=& \frac{3(n-k)(2d-k+1)}{6dn-3nk+3n-k^2+1}\,.
% \end{eqnarray}
\begin{eqnarray}
  \label{eq:computation-PIR-rate}
R_{\rm MBR} &=& \frac{(n-k)\left((d-k)k + \frac{k(k+1)}{2}\right)}{(d-k)nk + \sum_{j=1}^{k}j(n-k+j)}\\
&=& \frac{(n-k)\left((d-k)k + \frac{k(k+1)}{2}\right)}{(d-k)nk + (n-k)\frac{k(k+1)}{2} + \frac{k(k+1)(2k+1)}{6}}\nonumber\\
&=& \frac{3(n-k)(2d-k+1)}{6dn-3nk+3n-k^2+1}\nonumber\,.
\end{eqnarray}
\end{proof}

\begin{remark}
  As a function of $n, k, B$, the PIR rate given in Theorem~\ref{thm:PIR-MBR} can be written as
  \begin{equation}
    \label{eq:computation-PIR-rate-2}
    R_{\rm MBR} = \frac{1 - \frac{k}{n}}{1 - \frac{k(k+1)(k-1)}{6 nB}}\,.
  \end{equation}
  Indeed, starting from Equation~\eqref{eq:computation-PIR-rate} we get
  \[
    R_{\rm MBR} = \frac{(n-k)B}{nB + \sum_{j=1}^k j(j-k)} = \frac{(n-k)B}{nB - \frac{k(k+1)(k-1)}{6}}\,,
  \]
  leading to the expected expression.
\end{remark}

\subsection{On the PIR rate}

\subsubsection{Comparison with the multi-file PIR scheme of Dorkson and Ng}
%\textcolor{red}{Cami: Not sure we want to keep this, if we do should be moved down to where we also compare to the RS rate and the colluding rate.}
%\BLUE{[Julien: to me this is a nice comparison, and, if everything is correct, our scheme has a much better rate (see Figure~\ref{fig:comparison})][Razan: I think it is good to keep this section, and I also think it is better as a separate section since the next one compares rates of PIR schemes that are not built on MBR codes, but this is comparison to previous work.]}

Dorkson and Ng in~\cite{DorksonN18} proposed a PIR scheme over PM-MBR codes in the context of \emph{multi-file} retrieval, \emph{i.e.}\ any set of $p \ge 1$ files $\bfX^{f_0}, \dots, \bfX^{f_{p-1}}$ can be simultaneously retrieved privately. In the current work, retrieving $p$ files remains possible by iterating the $1$-file PIR protocol $p$ times. Notice that this routine achieves the same PIR rate as the $1$-file PIR scheme.

In the general case, the PIR rate obtained in~\cite{DorksonN18} is $R' = \frac{pB}{dn}$, under the additional constraint that $n = pk + d$. We notice that $R'$ can be reformulated as follows:
\[
R' = \frac{n-d}{k} \cdot \frac{B}{nd} = \frac{n-d}{n}  \cdot \frac{B}{kd}\,.
\]
Assume that $k  \le d < n$, which is the case for non-degenerate PM-MBR codes. This implies that $\frac{n-d}{n}=1-\frac{d}{n} \le 1 - \frac{k}{n}$ and $\frac{B}{kd} < 1$, and therefore
\[
R ' < 1 - \frac{k}{n} < R_{\rm MBR}\,,
\]
where $R_{\rm MBR}$ is the PIR rate of the scheme we propose in the current work. We emphasize our improvement upon~\cite{DorksonN18} with the numerical and asymptotic analyses proposed in Figure~\ref{fig:comparison}.

% {\bf Case $p=1$.} The rate in \cite{DorksonN18} in the case of a single message ($p=1$) being retrieved is 
% \[
% R_{p=1} = \frac{k(2d-k+1)}{2d(k+d)}\,.
% \]
% We point out that in the regenerating code regime, we should always have $d>k$, while the authors compute their `optimal' rate for $d=k$. Regarding notation in \cite{DorksonN18}, we have $r=d=\alpha$ \BLUE{[do we need to point out this?]}

% \textcolor{red}{Cami: Setting $n=6,k=3,d=4$ we get a rate 21/31 (fix!). Changing to $k=2$, we get 84/123 (corrected!). This is better than the new paper in single-message case, which would respectively give $(r=d=\alpha, p=1)$ $R=7/24$ for $k=2$ and $R=9/28$ for $k=3$. For $n=40,k=15,d=20$ we get 0.67 vs. 0.28. Could make some plots with their rate added too. New paper rate w k=2, p=2,10 resp: 7/16, 35/48.}

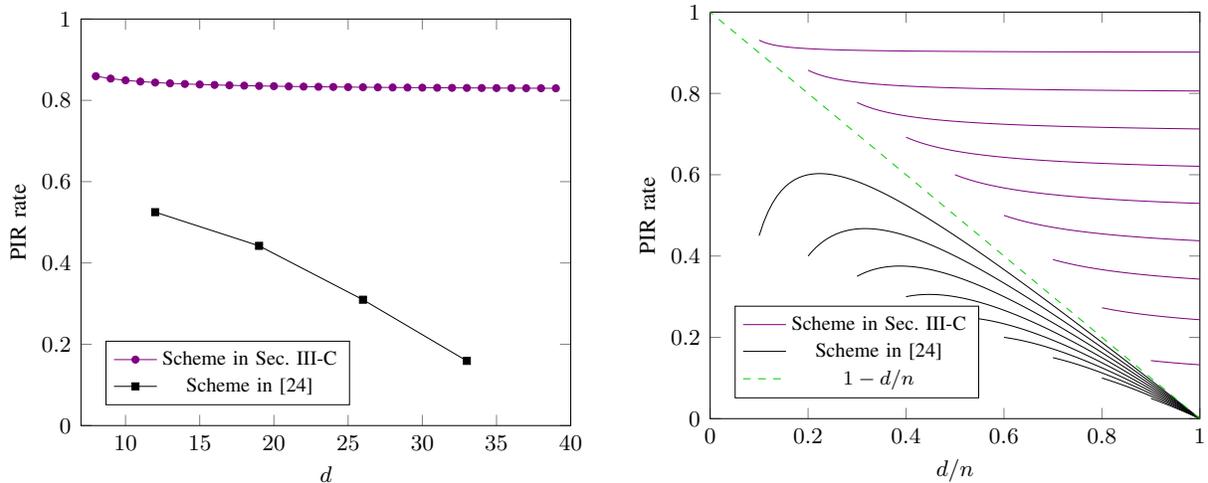
\begin{figure}
  \centering
  % \includegraphics[scale=0.2]{plot.png}
  % \begin{subfig}[t]{0.48\textwidth}
  \subfloat[
  PIR rate of both schemes, with a finite number of nodes $n$. We here set $n = 40$ and $k=7$, and we plot the PIR rate versus $d$. For fixed values of $n$ and $k$ and varying $k+1 \le d \le n-1$, the scheme in~\cite{DorksonN18} allows only a few admissible values of $p$, since $n = pk+d$ must hold. The larger the $p$, the larger the PIR rate of~\cite{DorksonN18}, but it remains bounded by the present scheme for every admissible value of $p$.
  ]
  {
    \begin{tikzpicture}[scale=0.95]
      
        \def\n{40.0}
        \def\k{7.0}

        \pgfplotsset{every tick label/.append style={font=\footnotesize}}
        
        \begin{axis}[
          xmin=\k,
          xmax=\n,
          ymin=0,
          ymax=1,
          xlabel={$d$},
          ylabel={PIR rate},
          xtick={5, 10, ..., 40},
          xlabel style={anchor = north, at={(0.5,-0.08), font=\footnotesize}},
          ylabel style={anchor = north, at={(-0.16,0.5), font=\footnotesize}},
          legend style= {anchor = south west, at={(0.05,0.05), font=\scriptsize}},
          cycle list name=mark list*,
          % yscale=0.8
          ]
          
          \tikzstyle{my_style}=[domain={\k+1}:{\n-1}, mark options={scale=0.7}, samples=\n-\k-1]

          \addplot+[my_style, color=violet]
          % [domain=\k+1:\n-1, color=violet!50!black, mark=square, mark options={scale=0.5, green!50!black}, samples=\n-\k-1]
          {(3*(\n-\k)*(2*x-\k+1))/(6*\n*x - 3*\n*\k + 3*\n - \k*\k + 1)}; 
          \addlegendentry{Scheme in Sec.~\ref{subsec:PIR-MBR}}

          \addplot+[my_style, color=black, domain=(\n-(ceil(\n/\k)-2)*\k):\n-\k, samples=ceil(\n/\k)-2]
          {(\n-x)*(2*x-\k+1))/(2*x*\n)}; 
          \addlegendentry{Scheme in~\begin{NoHyper}\cite{DorksonN18}\end{NoHyper}}
          
        \end{axis} 
      \end{tikzpicture}
    }
    ~ 
    %\begin{subfig}[t]{0.48\textwidth}
    \subfloat[PIR rate of both schemes, with an asymptotic number of nodes $n$. Each curve represents a distinct value of $k/n \in \{ 0.1, \dots, 0.9 \}$, and we plot the PIR rate versus $d/n$.]{
      \begin{tikzpicture}[scale=0.95]
        \pgfplotsset{every tick label/.append style={font=\footnotesize}}
        
        \begin{axis}[
          xmin=0,
          xmax=1,
          ymin=0,
          ymax=1,
          xlabel={$d/n$},
          ylabel={PIR rate},
          xlabel style={anchor = north, at={(0.5,-0.08), font=\footnotesize}},
          ylabel style={anchor = north, at={(-0.16,0.5), font=\footnotesize}},
          legend style= {anchor = south west, at={(0.05,0.05), font=\scriptsize}},
          cycle list name=mark list*,
          % yscale=0.8
          ]

          \tikzstyle{my_style}=[mark=none, samples=100]
          
          \foreach \g in {0.1,0.2,...,0.8} {
            \addplot+[my_style, color=violet, domain=\g:1, forget plot]
            {(3*(1-\g)*(2*x-\g))/(6*x - 3*\g - \g*\g)};         
          }
          \def\g{0.9}
          \addplot+[my_style, color=violet, domain=\g:1]
          {(3*(1-\g)*(2*x-\g))/(6*x - 3*\g - \g*\g)}; 
          \addlegendentry{Scheme in Sec.~\ref{subsec:PIR-MBR}}

          \foreach \g in {0.1,0.2,...,0.8} {    
            \addplot+[my_style, color=black, domain=\g:1, forget plot]
            {(1-x)*(2*x-\g))/(2*x)}; 
          }   
          \def\g{0.9} 
          \addplot+[my_style, color=black, domain=\g:1]
          {(1-x)*(2*x-\g))/(2*x)}; 
          \addlegendentry{Scheme in~\begin{NoHyper}\cite{DorksonN18}\end{NoHyper}}
          
          \addplot+[my_style, color=green!80!black, dashed]
          {1-x}; 
          \addlegendentry{$1 - d/n$}

        \end{axis} 
      \end{tikzpicture}
}
    %\end{subfig}
  \caption{\label{fig:comparison}Comparison between PIR rates of the multi-file PIR scheme in~\cite{DorksonN18} and the PIR scheme in the present paper.
 % \RED{Cami: Explain the role of $p$ here and when their system leads to better rates, if it does.}
  }
\end{figure}

\subsubsection{Comparison with the asymptotic capacities of scalar MDS codes}

Since PM-MBR codes allow to retrieve files by contacting only $k$ nodes among $n$, it is somewhat relevant to compare the proposed scheme with PIR schemes over $[n, k]$ MDS-coded data. We can also motivate this comparison by the following example.
\begin{example}
  In the PIR scheme presented in Example~\ref{ex:PIR-MBR}, the queried file has size $(n-k)B = 27$, while the user needs to download $18 + 18 + 10 + 4 = 50$ symbols. Hence, the PIR rate is $27/50$, which is larger than $1-k/n=1/2$, the PIR capacity of an $[n,k]$ MDS code, but smaller than $1-B/nd$, the PIR capacity of an $[nd,B]$ MDS code.
  
  However, in the MBR construction, $d$ symbols are stored on a single server. Therefore, considering the storage code as an $[nd, B]$ linear code, a PIR protocol must resist to \emph{some} sets of colluding nodes of size $d$ (also known as \emph{partial collusion}). In this setting, we can compare our construction to the conjectured PIR capacity $1 - \frac{B+d-1}{nd}$ of $[nd, B]$ linear codes with \emph{full} $d$-collusion~\cite{FreijGHK17}. In the current example, the conjectured capacity is then $1/2$, which is again below the achieved rate.
  %If each of these symbols is assumed to be stored on a separate server, the linear $[nd, B]$ code has a PIR capacity bounded by $1-B/nd$ if the servers do not collude.
  % This capacity, however, if we look at this code as an $[nd,b]$ code, then it is a code with $d$ colluding servers which can not achieve the $1-B/nd$ rate. \textcolor{blue}{Say that it is probably not achievable since those are vector codes..}
  % PM-MBR code is not an $[nd, B]$ MDS code. \RED{Cami: non-MDS codes can achieve the MDS-capacity as we know, remove this statement and rather refer to the fact that we have partial collusion.} %\BLUE{[todo: explain more]}
\end{example}

\begin{lemma}
  The PIR rate $R_{\rm MBR}$ of the scheme from Theorem~\ref{thm:PIR-MBR} satisfies:
  \[
  1 - \frac{k}{n} \le R_{\rm MBR} \le 1 - \frac{B}{nd}\,.
  \]
\end{lemma}
\begin{proof}
  If $1 \le j \le k$, it is clear that $n - k + j \le n$. Using this trivial observation in Equation~\eqref{eq:computation-PIR-rate}, we get 
\[
R_{\rm MBR} \ge \frac{(n-k)\left((d-k)k + k(k+1)/2\right)}{n(d-k)k + n\sum_{j=1}^{k}j} = \frac{n-k}{n} = 1 - \frac{k}{n}\,.
\]

The right-hand-side inequality is a bit more technical to state. %\BLUE{[if you guys have a simpler proof...]}
Using the expression of $R_{\rm MBR}$ given in Theorem~\ref{thm:PIR-MBR}, it is equivalent to prove that
\[
\Delta \mydef (nd - B)(6dn-3nk+3n-k^2+1) - 3(n-k)(2d-k+1)nd
\]
is non-negative. A computation shows that:
\[
\begin{aligned}
  2 \Delta &= (2nd -2kd+k^2-k)(6nd-3nk+3n-k^2+1) - 6nd(n-k)(2d-k+1) \\
           &= 6nd((2nd-2kd+k^2-k) - (n-k)(2d-k+1)) - (2nd-2kd+k^2-k)(k^2-1+3nk-3n)\\
           &= 6n^2d(k-1)    - (2nd-2kd+k^2-k)(k-1)(3n+k+1) \\
           &= (k-1)[6n^2d-(2nd-2kd+k^2-k)(3n+k+1)]\,.
\end{aligned}
\]

If $k=d$, then we get  $2 \Delta = k(k-1)(k+1)(n - (k+1)) \ge 0$ as long as $n \ge k+1$ which must hold for non-degenerated MBR codes.

If $d \ge k+1$, as it is for a non-trivial regenerating code, then we get
\[
\begin{aligned}
  \frac{2 \Delta}{k-1} &= d((k-1)(4n+2k+3) + 2n+2) - (k-1)(k+1)(3n+k+1) \\
           &\ge (k+1)((k-1)(4n+2k+3) + 2n+2) - (k-1)(k+1)(3n+k+1) \\
           &\ge (k+1)(k-1)(n+k+2) + 2(k+1)(n+1)\\
           &\ge 0\,.
\end{aligned}
\]
\end{proof}

%\GREEN{Ragnar: Can we simplify this argument, using less formulas? Let's at least think about it for the revision round.}
%\BLUE{[Julien: I tried the first time, but didn't manage]}\BLUE{[Razan: I tried too and couldn't.. :/]}

%\textcolor{red}{Cami:Here's a short rate comparison for our previous rate $1-k/n$. Please repeat this exercise to show when $R_{MBR}>R_*$ (after simplifying $R_{MBR}$).\\
%(should define the rates below more carefully...)}

% \BLUE{[Julien: here I don't remember why we wanted to compare with $1-k/n$. In Figure~\ref{fig:aweful-formula} I give formal solutions to the equation $R = 1 - \frac{B+d-1}{nd}$, where $R$ is our rate. They are quite awful. Maybe a numerical analysis would be better. See also Figure~\ref{fig:comparison-bounds}]}

We can also model the $n$ servers storing $\alpha=d$ symbols each as an $nd$-tuple of ``virtual'' or ``sub''-servers storing one symbol each. In this setting, some $d$-tuples of servers collude with one another. For that reason, it is relevant to compare the PIR rate of this scheme with the (conjectured) capacity of a PIR scheme for an $[nd, B]$ MDS-coded storage system allowing collusions of servers of size up to $\alpha = d$.
This conjectured capacity is $1- \frac{B+d-1}{nd}$~\cite{FreijGHK17}.
% Let us calculate next when the RS rate $R_{RS} = 1-k/n$ is larger than the coded colluded capacity $R_*$ of an $(n\alpha,B)$ linear code, \emph{i.e.}, 
% $$
% 1-\frac{k}{n}>1-\frac{B+\alpha-1}{n\alpha}\,.
% $$
% \GREEN{Razan: I am not quite sure what the significance of this comparison is. The comparisons in the figure, that compares all those rates to our rate are important, but this specific comparison, I do not really see the point of it. I mean this does not even compare our rate.}
Note that the assumption of full $d$-collusion is pessimistic since in this setting, not any $d$ servers can collude, rather there exist disjoint sets of colluding servers that are known a priori, cf. \cite{tajeddine2017private}. %\textcolor{red}{Cami: Btw, what is we assume (full) $t$-collusion in addition to this a priori disjoint $\alpha$ collusion? Should we try to address that as well?}

A comparison of the rate of the PIR scheme constructed in this paper with the other relevant capacity expressions of PIR schemes discussed in this section is shown in Figure~\ref{fig:comparison-bounds} for different values of $n,k$ and $d$. We can see that the achieved rate in our scheme is higher than the PIR capacity of an $[n,k]$ MDS code, and for a reasonably high value of $d$, the achievable PIR rate for the scheme described in Section~\ref{subsec:PIR-MBR}. As explained before, the achieved rate is always lower than the PIR capacity of an $[nd,B]$ MDS code. %\RED{cami: And...? Reflect the results.}

\begin{figure}
  \centering
  % \includegraphics[scale=0.2]{plot.png}
  % \begin{subfig}[t]{0.48\textwidth}
    \subfloat[PIR rate versus $d$ when $n=40$ and fixed $k=7$.]{
      \begin{tikzpicture}[scale=0.95]
      
      \def\n{40.0}
      \def\k{7.0}

      \pgfplotsset{every tick label/.append style={font=\footnotesize}}
      
      \begin{axis}[
        xmin=\k,
        xmax=\n,
        ymin=0.78,
        ymax=0.92,
        xlabel={$d$},
        ylabel={PIR rate},
        xtick={5, 10, ..., 40},
        xlabel style={anchor = north, at={(0.5,-0.08), font=\footnotesize}},
        ylabel style={anchor = north, at={(-0.16,0.5), font=\footnotesize}},
        legend style= {anchor = north east, at={(0.95,0.95), font=\scriptsize}},
        cycle list name=mark list*,
        % yscale=0.8
        ]      

        \tikzstyle{my_style}=[domain={\k+1}:{\n-1}, mark options={scale=0.7}, samples=\n-\k-1]

        \addplot+[my_style, color=violet]
        % [domain=\k+1:\n-1, color=violet!50!black, mark=square, mark options={scale=0.5, green!50!black}, samples=\n-\k-1]
        {(3*(\n-\k)*(2*x-\k+1))/(6*\n*x - 3*\n*\k + 3*\n - \k*\k + 1)}; 
        \addlegendentry{Scheme in Sec.~\ref{subsec:PIR-MBR}}
        
        \addplot+[my_style, color=blue!80!black]
        % [domain=\k+1:\n-1, color=blue!50!black, mark=square, mark options={scale=0.5, blue!50!black}, samples=\n-\k-1]
        {1-\k/\n};      
        \addlegendentry{$1 - \frac{k}{n}$}

        \addplot+[my_style, color=orange!80!black]
        % [domain=\k+1:\n-1, color=orange!50!black, mark=square, mark options={scale=0.5, orange!50!black}, samples=\n-\k-1]
        {1 - ((\k*(\k+1)/2 + \k*(x-\k))/(\n*x))}; 
        \addlegendentry{$1 - \frac{B}{dn}$}

        \addplot+[my_style, color=red]
        % [domain=\k+1:\n-1, color=orange!50!black, mark=square, mark options={scale=0.5, orange!50!black}, samples=\n-\k-1]
        {1 - ((\k*(\k+1)/2 + \k*(x-\k))/(\n*x)) - (x-1)/(\n*x)}; 
        \addlegendentry{$1 - \frac{B+(d-1)}{dn}$}

      \end{axis} 
    \end{tikzpicture}
    }
    % \caption{PIR rate versus $d$ when $n=40$ and fixed $k=7$.}
    % \end{subfig}
    ~
    % \begin{subfig}[t]{0.48\textwidth}
    \subfloat[PIR rate $R_{\rm MBR}$ versus $d$ when $n = 40$, assuming $d = 2(k-1)$.]{  
    \begin{tikzpicture}[scale=0.95]
        
        \def\n{40.0}
        \def\dmin{2}
        \def\dmax{\n}

        \pgfplotsset{every tick label/.append style={font=\footnotesize}}
        
        \begin{axis}[
          xmin=\dmin,
          xmax=\dmax,
          ymin=0.4,
          ymax=1,
          xlabel={$d$},
          ylabel={PIR rate},
          xtick={5, 10, ..., 40},
          xlabel style={anchor = north, at={(0.5,-0.08), font=\footnotesize}},
          ylabel style={anchor = north, at={(-0.16,0.5), font=\footnotesize}},
          legend style= {anchor = south west, at={(0.05,0.05), font=\scriptsize}},
          cycle list name=mark list*,
          ]      

          \tikzstyle{my_style}=[domain={\dmin+2}:{\n-2}, mark options={scale=0.7}, samples={(\dmax-\dmin)/2-1}]

          \addplot+[my_style, color=violet]
          % [domain=\k+1:\n-1, color=violet!50!black, mark=square, mark options={scale=0.5, green!50!black}, samples=\n-\k-1]
          {(3*(\n-(x/2+1))*(2*x-(x/2+1)+1))/(6*\n*x - 3*\n*(x/2+1) + 3*\n - (x/2+1)*(x/2+1) + 1)}; 
          \addlegendentry{Scheme in Sec.~\ref{subsec:PIR-MBR}}
          
          \addplot+[my_style, color=blue!80!black]
          % [domain=\k+1:\n-1, color=blue!50!black, mark=square, mark options={scale=0.5, blue!50!black}, samples=\n-\k-1]
          {1-(x/2+1)/\n};      
          \addlegendentry{$1 - \frac{k}{n}$}

          \addplot+[my_style, color=orange!80!black]
          % [domain=\k+1:\n-1, color=orange!50!black, mark=square, mark options={scale=0.5, orange!50!black}, samples=\n-\k-1]
          {1 - (((x/2+1)*((x/2+1)+1)/2 + (x/2+1)*(x-(x/2+1)))/(\n*x))}; 
          \addlegendentry{$1 - \frac{B}{dn}$}

          \addplot+[my_style, color=red]
          % [domain=\k+1:\n-1, color=orange!50!black, mark=square, mark options={scale=0.5, orange!50!black}, samples=\n-\k-1]
          {1 - (((x/2+1)*+((x/2+1)+1)/2 + (x/2+1)*(x-(x/2+1)))/(\n*x)) - (x-1)/(\n*x)}; 
          \addlegendentry{$1 - \frac{B+(d-1)}{dn}$}
        \end{axis} 
      \end{tikzpicture}
      % \caption{PIR rate $R_{\rm MBR}$ versus $d$ when $n = 40$, assuming $d = 2(k-1)$.}
    }
  \caption{\label{fig:comparison-bounds}PIR rate versus $d$ when $n=40$. Comparison between the PIR rate of the scheme in the current work, the capacity $1- k/n$ of a PIR scheme for an $[n,k]$ MDS coded storage system with no collusion, the capacity $1- B/nd$ of a PIR scheme for an $[nd, B]$ MDS coded storage system with no collusion, and the conjectured capacity $1 - (B+d-1)/nd$ of a PIR scheme for an $[nd,B]$ MDS-coded storage system with full $d$ collusion.}
\end{figure}
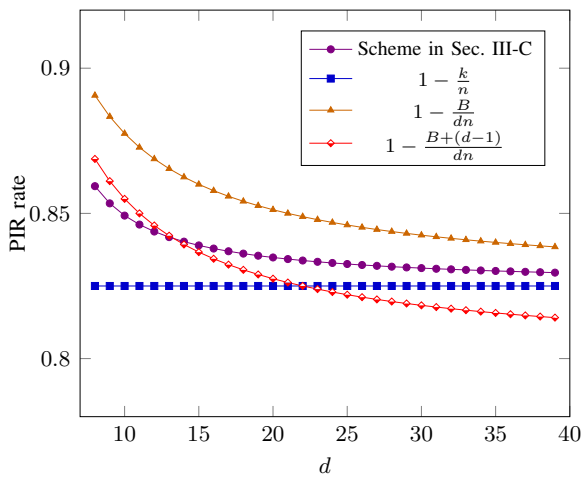
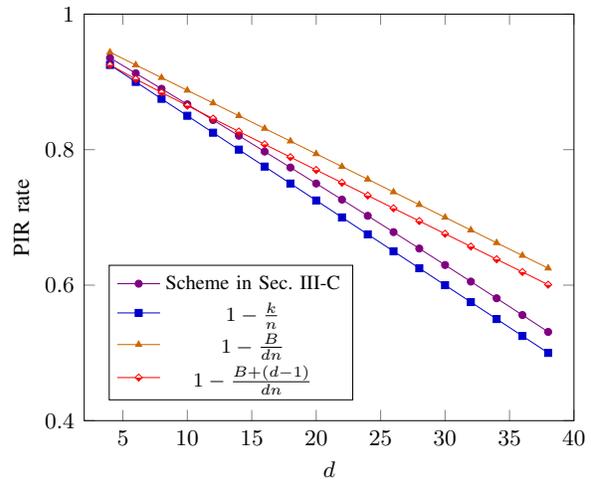

\section{A PIR scheme in the MSR setting}
\label{sec:PIR-MSR}

We consider a regenerating code $\calC$ attaining the MSR point. As explained in Section~\ref{subsubsec:def-MSR}, we restrict our work on the setting $d = 2k-2 = 2 \alpha$ for simplicity. Hence $\calC$ is also a linear code over $\FF_q$ of length $nd=2n\alpha$ and dimension $B = \alpha(\alpha+1)$. 

\subsection{System setup}

 Similarly to the MBR setting, we consider a storage system $\bfX$ of $F$ files $\bfX^1, \dots, \bfX^F$, each storing $B = \alpha(\alpha+1)$ information symbols. The symbols of the file $\bfX^f$, $1 \le f \le F$, are arranged into $S= n-2\alpha$ stripes, such that the message $\bfM^f$ can be written
\[
\bfM^f = \left( M^f[i,j,s],
  {
    \scriptsize
    \begin{array}{l}
      1 \le i \le 2 \alpha \\
      1 \le j \le \alpha \\
      1 \le s \le S
    \end{array}
  }
\right).
\]
By construction of the MSR code $\calC$, for all $1 \le f \le F$ and all $1 \le i,j \le \alpha$, we have 
\[
\bfM^f[i, j, \bcdot] = \bfM^f[j, i ,\bcdot] = \bfM^f[\alpha + i, j, \bcdot]\,.
\]
Moreover, for every $j, s, f$, the column $\bfM^f[\bcdot, j, s] \in \FF_q^{2\alpha}$ is encoded into a Reed-Solomon codeword $\bfC^f[\bcdot, j, s] \in \RS_{2\alpha}(\bfx)$ by
\[
 \bfC^f[\bcdot, j, s] = \sum_{r=1}^{2\alpha} M^f[r, j, s] \bfgamma_r,
\]
where we recall that $\{ \bfgamma_1, \dots, \bfgamma_{2\alpha} \}$ denotes a suitable basis for sequences of Reed-Solomon codes (see Section~\ref{subsec:notation}).

\subsection{The PIR scheme}
\label{subsec:PIR-MSR}

Assume the user wants to retrieve file $\bfX^{f_0}$ privately. We consider a $2\alpha$-tuple of queries $\bfQ = (\bfQ_1, \dots, \bfQ_{2\alpha})$ having the following form for $1 \le \ell \le 2\alpha$:
\[
\bfQ_\ell =  \left( Q^f_\ell[i,s],
  {
    \scriptsize
    \begin{array}{l}
      1 \le i \le n \\
      1 \le s \le S\\
      1 \le f \le F
    \end{array}
  }
\right).
\]
Once again, $\bfQ_\ell$ does \emph{not} depend on the column index $j \in [1, \alpha]$, preventing to leak information on the requested file.
\medskip

  {\bf Generation of $\bfQ$.} Similar to the MBR setting, queries $\bfQ$ are defined by $\bfQ \mydef \bfD + \bfE^{(f_0)}$ with $\bfD$ and $\bfE^{(f_0)}$ defined as follows.
\begin{enumerate}
  \item For every $\ell, s, f$, the random vector $\bfD^f_\ell[\bcdot, s] \in \FF_q^n$ is a word picked uniformly at random from the repetition code of length $n$.
  \item The retrieval pattern $\bfE^{(f_0)}$ is defined by
    \[
    E^{(f_0), f}_\ell[i, s] = 
    \left\{
      \begin{array}{ll}
        1 &\text{if } f=f_0 \text{ and } n-i = \ell + s - 2\quad (\text{mod } S),\\
        0 &\text{otherwise,}
      \end{array}
    \right.
    \]
    for every $1 \le \ell \le 2\alpha$, $1 \le i \le n$, $1 \le s \le S$ and $1\le f \le F$.
\end{enumerate}
\medskip

  {\bf Server responses to queries.} Given a column $1 \le j \le \alpha$, only servers $S_i$ such that $2\alpha-2j+1 \le i \le n$ are required to send the subset of responses $R_\ell[i,j]$, for $1 \le \ell \le 2j$. 
\medskip

  {\bf Reconstruction of $\bfX^{f_0}$.} The recovery is run columnwise, from column $\alpha$ down to $1$. In every step $1 \le j \le \alpha$, the goal is to retrieve $\bfM^{f_0}[\bcdot, j, \bcdot]$ as well as some random vectors. The recovery procedure is identical to that of the first columns of the MBR case. Column $\alpha$ is retrieved using a classical PIR protocol on MDS codes, as in \cite{tajeddine2016private}. Here, the underlying storage code is $\RS_{2\alpha}(\bfx)$. Similarly to the MBR case, the user retrieves pieces of the required file, along with some randomness. The collected symbols from column $\alpha$ (randomness and information symbols) can be reused in column $\alpha-1$ to again retrieve other pieces of the required file and associated randomness. This process is then repeated until retrieving the information from column $1$. This iterative process reduces the number of total downloaded symbols to retrieve the required file $\bfX^{f_0}$, and consequently reduces the PIR rate. We refer to Lemma~\ref{lem:correctness-MSR} for technical details.
\medskip

We give a simple example to explain the scheme. 

\begin{example}
  \label{ex:PIR-MSR}
  We use the $(6, 3, 4)$ PM-MSR regenerating code presented in Example~\ref{ex:MSR}, with $\alpha = 2$. Files are divided into $S=n-2\alpha=2$ stripes, and the user sends $2\alpha=4$ vectors of queries:
  \[
  \centering
  \begin{array}{|c|c|c|c|c|}
    \hline
                       & \text{Query } 1   & \text{Query } 2   & \text{Query } 3    & \text{Query } 4   \\\hline
    \text{Server } S_1 & \bfu              & \bfv              & \bfw               & \bfy              \\\hline
    \text{Server } S_2 & \bfu              & \bfv              & \bfw               & \bfy              \\\hline
    \text{Server } S_3 & \bfu  & \bfv  & \bfw + \bfe_{f_0,1} & \bfy + \bfe_{f_0,2} \\\hline
    \text{Server } S_4 & \bfu  & \bfv  & \bfw + \bfe_{f_0,2} & \bfy + \bfe_{f_0,1}  \\\hline
    \text{Server } S_5 & \bfu + \bfe_{f_0,1} & \bfv + \bfe_{f_0,2} & \bfw  & \bfy  \\\hline
    \text{Server } S_6 & \bfu + \bfe_{f_0,2} & \bfv + \bfe_{f_0,1} & \bfw  & \bfy \\\hline
  \end{array}
  \]  
  The vector $\bfe_{f_0,s_0} \in \FF_q^{F \times (n-2\alpha)}$ is the all zero vector with a single $1$ in position $(f_0,s_0)$, \emph{i.e.}, indicating stripe $s_0$ from file $\bfX^{f_0}$. Vectors $\bfu, \bfv, \bfw, \bfy \in \FF_q^{F \times S}$ are random vectors.
  
  The servers project the data stored in column $2$ on all the queries. Servers $S_1$ and $S_2$ do not respond to any other queries. Servers $S_3,\dots, S_6$ project only the first $2$ queries on the data stored in the first column.
\end{example}

\subsection{Proofs}

  For $1 \le j \le \alpha$, we define the $2j$-dimensional code
  \[
  \calC_j \mydef \RS_j(\bfx) + \langle \bfx^\alpha \rangle \star \RS_j(\bfx) \subseteq \FF_q^n\,.
  \]

\begin{lemma}
  There exists a sequence $I_1 \subset \dots \subset I_\alpha \subset [1,n]$ such that, for every $1 \le j \le \alpha$, $I_j$ is an information set for  the code $\calC_j$.
\end{lemma}

\begin{proof}
  We prove the result inductively. First notice that $\calC_\alpha = \RS_{2\alpha}(\bfx)$, hence one can choose any $2\alpha$-subset for $I_\alpha$. Then, it is sufficient to notice that for every $2 \le j \le \alpha$, we have $\calC_{j-1} \subset \calC_j$. Hence, an information set $I_j$ for $\calC_j$ contains an information set for $\calC_{j-1}$.
\end{proof}

The previous lemma allows us to make the following assumption: after reordering the servers (\emph{i.e.} the evaluation points $\bfx$), we can assume that $I_j = [2\alpha - 2j +1, 2\alpha]$ for every $1 \le j \le \alpha$. Moreover, we define the code $\calA_j \subseteq \FF_q^{n-2\alpha + 2j}$ as the puncturing of $\calC_j$ on its $(2\alpha - 2j)$ first coordinates. The code $\calA_j$ has length $n-2\alpha + 2j$ and dimension $2j$, and by the chosen order of coordinates, its $2j$ first coordinates form an information set.

%\GREEN{Ragnar: Again, in this lemma, I would like more mathematical precision, both in the formulation and more importantly in the proof. See inline comments.}
\begin{lemma}
  \label{lem:correctness-MSR}
  Let $1 \le j \le \alpha$. For every $1 \le \ell \le 2j$, we denote \[\bfR_\ell[\bcdot, j] \mydef (R_\ell[2\alpha-2j+1,j], \dots, R_\ell[n,j]) \in \FF_q^{n-2\alpha+2j}.\] Then, conditioned on $(\bfR_1[\bcdot, j], \dots, \bfR_{2j}[\bcdot, j])$ and on 
  \begin{equation}
    \label{eq:random-vectors-correctness-MSR}
  \sum_{s,f} M^f[r, j, s] \bfD_\ell^f[\bcdot, s] \quad, \text{for all } \; j+1 \le r \le \alpha, \quad 1 \le \ell \le 2j\,,
  \end{equation}
  the following is determined:
  \begin{itemize}
  \item the piece $\bfM^{f_0}[\bcdot, j, \bcdot]$ of the desired file;
  \item the random vectors $\sum_{s,f} M^f[r,j,s] \bfD_\ell^f[\bcdot, s] \in \FF_q^n$, for all $1 \le r \le 2\alpha$ and every $1 \le \ell \le 2j$.
  \end{itemize}
\end{lemma}

\begin{proof}
  The proof is very similar to the one of Lemma~\ref{lem:correctness-first-columns}. Let us fix $1 \le \ell \le j$. 
  %Denote by $\bfx' = (x_{2\alpha-2j+1}, \dots, x_n)$.
  The user can build
  \[
    \bfR_\ell[\bcdot, j] \mydef (R_\ell[2\alpha-2j+1,j], \dots, R_\ell[n,j]) = \bfA_\ell[\bcdot, j] + \bfB_\ell[\bcdot, j]\,,
  \]
  where $\bfA_\ell[\bcdot, j]$ and $\bfB_\ell[\bcdot, j]$ are defined as in Lemma~\ref{lem:correctness-last-columns}.
  
  Denote $J_1 \mydef [0, j-1] \cup [\alpha, \alpha + j-1]$ and $J_2 \mydef [0, 2\alpha -1 ] \setminus J_1$. Both $J_1$ and $J_2$ are publicly known to the user and the servers, as they only depend on the parameters of the scheme.
  
  Define $\bfgamma_r \mydef (x_{2\alpha-2j+1}^r, \dots, x_n^r) \in \FF_q^{n-2\alpha+2j}$, for $0 \le r \le 2\alpha-1$. It is clear that $\{\bfgamma_r, r \in J_1 \}$ is a basis of the code $\calA_j$ defined above.  One can rewrite $\bfA_\ell[\bcdot, j] \in \FF_q^{n-2\alpha+2j}$ as follows:
  \[
  \begin{aligned}
    \bfA_\ell[\bcdot, j]
    &= \sum_{s,f} \bfD_\ell^f[\bcdot, s] \star \bfC^f[\bcdot, j, s]\\
    &= \sum_{s,f}  \bfD_\ell^f[\bcdot, s] \star \left(\sum_{r=1}^d M^f[r, j, s] \bfgamma_r \right)\\
    &= \sum_{r \in J_1} \sum_{s,f} M^f[r, j, s] \; \bfD_\ell^f[\bcdot, s] \star \bfgamma_r  + \sum_{r \in J_2} \sum_{s,f} M^f[r, j, s] \; \bfD_\ell^f[\bcdot, s] \star \bfgamma_r.
  \end{aligned}
  \]
  Therefore, using random vectors given in \eqref{eq:random-vectors-correctness-MSR},
  % \GREEN{Ragnar: Be precise, which random vectors and why? I guess $J_1$ is designed so that the user knows everything from it, but this is not clearly stated and also not immediately clear to me.}
  the vector
  \[
  \bfA'_\ell[\bcdot, j] \mydef \sum_{r \in J_2} \left( \sum_{s,f}  M^f[r, j, s] \; \bfD_\ell^f[\bcdot, s] \right) \star \bfgamma_r
  \]
  can be constructed by the user. Recall that for any file $X^f$,
  \[
  \bfM^f[r, j, \bcdot] = \bfM^f[j, r, \bcdot] = \bfM^f[\alpha+r, j, \bcdot]
  \]
  for every $1 \le r \le \alpha$. Hence, the user is able to construct
  \[
    \bfR''_\ell[\bcdot, j] \mydef \bfR_\ell[\bcdot, j] - \bfA'_\ell[\bcdot, j] =  (\bfA_\ell[\bcdot, j] - \bfA'_\ell[\bcdot, j]) + \bfB_\ell[\bcdot, j]
  \]
  and, by definition of $J_1$, we see that $\bfA''_\ell[\bcdot, j] \mydef \bfA_\ell[\bcdot, j] - \bfA'_\ell[\bcdot, j]$ lies in $\calA_j$. We remark that, once again, the vector $\bfB_\ell[\bcdot, j] \in \FF_q^{n-k+j}$ is supported  on $[2\alpha+1, n]$. According to the discussion preceding the lemma, the interval $I_j = [2\alpha-2j+1, 2\alpha]$ is an information set for $\calC_j$. Therefore the user can recover $\bfA''_\ell[\bcdot, j]$ and $\bfB_\ell[\bcdot, j]$ from $\bfR''_\ell[\bcdot, j]$.
  
  Finally, the recovery of $\bfM^{f_0}[\bcdot, j, \bcdot]$ and of random elements $\sum_{s,f} M^f[r,j,s] \lambda_{\ell, s, f}$ is identical to Lemma~\ref{lem:correctness-first-columns}.
\end{proof}

\begin{theorem}
  \label{thm:PIR-MSR}
  The scheme proposed in Section~\ref{subsec:PIR-MSR} is secure against non-colluding servers. Its PIR rate is
  \[
  R_{\rm MSR} = \frac{3(n - 2\alpha)}{3n - 2 \alpha + 2 }.
  \]
\end{theorem}

\begin{proof}
  We have seen in Lemma~\ref{lem:correctness-MSR} that the proposed scheme reconstructs the correct file. Similarly to the MBR case, the scheme is private if servers do not collude. Let us compute the PIR rate.

  The desired file consists of $(n - 2\alpha) B = \alpha(\alpha+1)(n - 2\alpha)$ symbols. For column $1 \le j \le \alpha$, the number of downloaded symbols is $2j \times (n - 2\alpha + 2j)$. Hence the PIR rate of the scheme is given by
\[
\begin{aligned}
R_{\rm MSR} &= \frac{\alpha(\alpha+1)(n - 2\alpha)}{\sum_{j=1}^\alpha2j (n - 2\alpha + 2j) }\\
&=\frac{\alpha(\alpha+1)(n - 2\alpha)}{n\alpha(\alpha+1) - 4 \sum_{j=1}^\alpha j  (\alpha - j) }\\
&=\frac{\alpha(\alpha+1)(n - 2\alpha)}{n\alpha(\alpha+1) - \frac{2}{3} \alpha(\alpha+1)(\alpha-1) }\\
&=\frac{3(n - 2\alpha)}{3n - 2 \alpha + 2 }\\
&= 1 - \frac{4\alpha +2}{3n - 2\alpha +2}\,.
\end{aligned}
\]
\end{proof}

\subsection{On the PIR rate in the MSR case}

In our simplified setting, it must hold that $\alpha = d/2 = k-1$. The PIR rate of the proposed scheme is then 
\[
R_{\rm MSR} = 1 - \frac{4\alpha +2}{3n - 2\alpha +2}\,.
\]
Dorkson and Ng~\cite{DorksonN18} give a multi-file PIR scheme for the same MSR codes, with a PIR rate of $\frac{\alpha(n-d)}{\alpha n} = 1 - d/n$. We prove in the following lemma that the PIR rate of our construction improves upon this rate.

\begin{lemma}
  Let $1 \le \alpha \le n/2$ and assume that $n \ge 6$ or $\alpha \ge 3$. Then:
  \[
  1 - \frac{d}{n} \le R_{\rm MSR} \le 1 - \frac{k}{n}\,.
  \]
\end{lemma}

\begin{proof}
  For the left-hand side inequality, we need to prove that $d/n \ge (2d+2)(3n-2d+2)$. A simple computation shows it is equivalent to $(d-2)(n-d) \ge 0$, which holds as long as $\alpha = d/2 \ge 1$.

  Similarly, the right-hand side inequality $R_{\rm MSR} \le 1 - \frac{k}{n}$ holds if and only if $\frac{n(\alpha-3)}{2} + (\alpha+1)(\alpha-1) \ge 0$, which proves our result.
\end{proof}

\begin{figure}
  \centering

  \begin{tikzpicture}
    
    \def\n{40.0}
    \def\amin{0}
    \def\amax{\n/2}

    \pgfplotsset{every tick label/.append style={font=\footnotesize}}
    
    \begin{axis}[
      xmin=\amin,
      xmax=\amax,
      ymin=0,
      ymax=1,
      xlabel={$\alpha = d/2 = k-1$},
      ylabel={PIR rate},
      xtick={5, 10, ..., 20},
      xlabel style={anchor = north, at={(0.5,-0.08), font=\footnotesize}},
      ylabel style={anchor = north, at={(-0.16,0.5), font=\footnotesize}},
      legend style= {anchor = south west, at={(0.05,0.05), font=\scriptsize}},
      cycle list name=mark list*
      ]      

      \tikzstyle{my_style}=[domain={\amin+1}:{\amax-1}, mark options={scale=0.7}, samples={\amax-\amin-1}]

      \addplot+[my_style, color=violet]
      % [domain=\k+1:\n-1, color=violet!50!black, mark=square, mark options={scale=0.5, green!50!black}, samples=\n-\k-1]
      {1 - (4*x+2)/(3*\n-2*x+2)}; 
      \addlegendentry{Scheme in Sec.~\ref{subsec:PIR-MSR}}

      \addplot+[my_style, color=green!60!black]
      % [domain=\k+1:\n-1, color=blue!50!black, mark=square, mark options={scale=0.5, blue!50!black}, samples=\n-\k-1]
      {1-(2*x)/\n};      
      \addlegendentry{Scheme in~\begin{NoHyper}\cite{DorksonN18}\end{NoHyper}: $1 - \frac{d}{n}$}
      
      \addplot+[my_style, color=blue!80!black]
      % [domain=\k+1:\n-1, color=blue!50!black, mark=square, mark options={scale=0.5, blue!50!black}, samples=\n-\k-1]
      {1-(x+1)/\n};      
      \addlegendentry{$1 - \frac{k}{n}$}

      % \addplot+[my_style, color=orange!80!black]
      % % [domain=\k+1:\n-1, color=orange!50!black, mark=square, mark options={scale=0.5, orange!50!black}, samples=\n-\k-1]
      % {1 - (x+1)/(2*\n))}; 
      % \addlegendentry{$1 - \frac{B}{2 \alpha n}$}

      % \addplot+[my_style, color=red]
      % % [domain=\k+1:\n-1, color=orange!50!black, mark=square, mark options={scale=0.5, orange!50!black}, samples=\n-\k-1]
      % {1 - (x*(x+1) - (x-1))/(2*x*\n)}; 
      % \addlegendentry{$1 - \frac{B-(\alpha-1)}{2 \alpha n}$}

    \end{axis} 
  \end{tikzpicture}
  \caption{PIR rate versus $\alpha$ in the MSR case, for fixed $n=40$. Recall that $2\alpha = d = 2k-2$ must hold.}  
  \label{fig:comparison-MSR}
\end{figure}
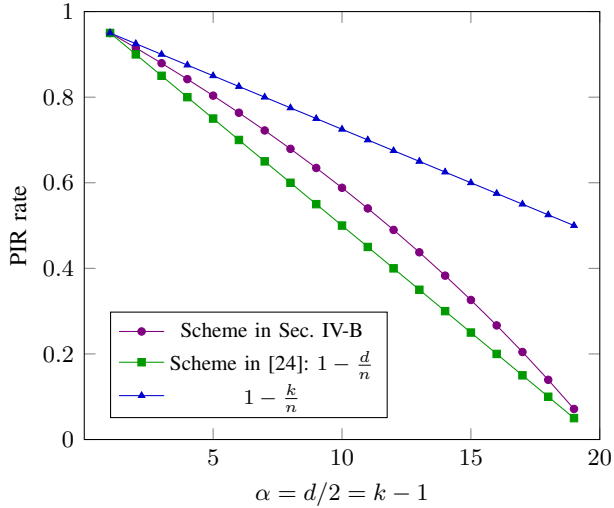

\section{Conclusion}

In this paper, we construct PIR schemes for the product matrix constructions in the MBR and MSR settings. The schemes use the symmetric properties of the PM codes in order to increase the PIR rate. For the PM-MBR setting, we achieve a PIR rate that is better than $1-k/n$, \emph{i.e.}, larger than the PIR capacity of an $[n,k]$ MDS coded storage system. As for the PM-MSR setting, we achieve a PIR rate between $1-d/n$, \emph{i.e.}, the PIR capacity of an $[n,d]$ MDS code, and $1-k/n$.

A possible further work on the topic would be to consider colluding servers. A natural idea is to adapt the constructions of Freij-Hollanti \emph{et al.}~\cite{FreijGHK17, freij2018t}, by replacing the repetition code where random vectors $\bfD^f_\ell[\bcdot, s]$ are picked, by a Reed-Solomon code of higher dimension. However, the extraction of the randomness --- necessary to decrease the communication 
cost of our schemes --- cannot be done as easily as in the non-colluding case, because projected random symbols interfere with themselves.

 \section*{Acknowledgments}
  The work of J. Lavauzelle is partially funded by French ANR-15-CE39-0013-01 \enquote{Manta}.\\
  The work of R. Tajeddine and C. Hollanti is supported in part by the Academy of Finland, under grants \#276031, \#282938, and \#303819 to C.~Hollanti, and by the Technical University of Munich -- Institute for Advanced Study, funded by the German Excellence Initiative and the EU 7th Framework Programme under grant agreement \#291763, via a \emph{Hans Fischer Fellowship} held by C.~Hollanti.\\
  The work of R. Freij-Hollanti is supported by the German Research Foundation (Deutsche Forschungsgemeinschaft, DFG) under Grant WA3907/1-1. 

%\RED{cami: In bib, ref [13] is missing info. Capitalization is not consistent, [22] not capitalized.} 
\bibliographystyle{ieeetr}
\bibliography{coding2}

\end{document}